\documentclass[aps,pra,longbibliography,superscriptaddress,nofootinbib,twocolumn, 10pt, floatfix]{revtex4-1}

\newif\ifarxiv
\newif\ifjournal
\journalfalse \arxivtrue  

\ifjournal
\usepackage{bibunits}
\defaultbibliography{sample_complexity_supremacy.bib}
\defaultbibliographystyle{./myapsrev4-1}
\fi

\usepackage[T1]{fontenc}
\usepackage[utf8]{inputenc}

\usepackage{float}
\usepackage{graphicx}
\usepackage{epsfig,epstopdf}
\usepackage[table]{xcolor}
\usepackage{tikz}
\usepackage{comment}
\usepackage[normalem]{ulem}

\usepackage[english]{babel}

\PassOptionsToPackage{hyphens}{url}
\usepackage[hyphenbreaks]{breakurl}

\usepackage{times}

\usepackage{amsmath}
\usepackage{bbm}
\usepackage{amsfonts}
\usepackage{amssymb}
\usepackage{amsthm}
\usepackage{mathbbol}
\usepackage{mathtools}
\usepackage{complexity}
\newclass{\classP}{P}
\newclass{\np}{NP}
\newclass{\sharpP}{\#P}
\newclass{\BQP}{BQP}

\usepackage{float}
\usepackage{graphicx}
\usepackage{epsfig,epstopdf}
\usepackage{xcolor}

\usepackage{soul}

\usepackage{enumerate}

\usepackage{pgfplots}
\usetikzlibrary{pgfplots.groupplots}
\pgfplotsset{compat=1.11}
\usepgfplotslibrary{fillbetween}

\usepackage{tikz}
\usepackage{pgffor}
\usetikzlibrary{
  shapes,
  shapes.geometric,
	trees,
	matrix,
  positioning,
  fit,
  shadings,
    pgfplots.groupplots,
  }
\definecolor{newblue}{RGB}{40,210,251}
\definecolor{lightgray}{RGB}{170,170,170}
\definecolor{darkyellow}{RGB}{255,210,70}
\definecolor{darkyellow2}{RGB}{251,184,38}
\definecolor{metalblue}{RGB}{78,156,219}
\definecolor{metalblue2}{RGB}{34,52,103}
\definecolor{pink}{RGB}{237,16,118}
\definecolor{pink2}{RGB}{131,28,71}
\definecolor{violet}{HTML}{53257F} 
\definecolor{violet2}{RGB}{61,18,100}
\definecolor{applegreen}{rgb}{0.55, 0.71, 0.0}
\definecolor{applegreen2}{rgb}{0.4, 0.8, 0.0}

\definecolor{DarkGray}{gray}{0.25} 
\definecolor{MidGray}{gray}{0.38} 
\definecolor{NeutralGray}{gray}{0.5}
\definecolor{LightGray}{gray}{0.7}
\definecolor{lightGray}{gray}{0.85}
\definecolor{DarkRed}{rgb}{0.7,0,0}
\definecolor{DarkBlue}{rgb}{0,0,0.5}
\definecolor{SteelBlue}{rgb}{0,0.4,0.6}
\definecolor{Orange}{rgb}{0.7,0.5,0}
\definecolor{Violette}{rgb}{0.5,0,0.5}
\definecolor{Sand}{rgb}{0.84,0.8,0.55}
\definecolor{niceblue}{rgb}{0.33,0.5,0.8}
\definecolor{OliveGreen}{RGB}{0,102,102}
\definecolor{NiceGreen}{RGB}{0,153,72}

\colorlet{tensorcol}{niceblue!70!gray}
\definecolor{changecol}{rgb}{0.7,0,0}

\colorlet{tensorcol1}{darkyellow}
\colorlet{tensorcol1border}{darkyellow2}
\colorlet{tensorcol2}{metalblue}
\colorlet{tensorcol2border}{metalblue2}

\colorlet{tensorcol}{niceblue!70!gray}
\definecolor{changecol}{rgb}{0.7,0,0}
\definecolor{applegreen}{rgb}{0.55, 0.71, 0.0}

\tikzset{
	ol/.style = {remember picture, overlay},
	eq-pic/.style = {inner sep = .5pt,draw, very thick, gray, rounded corners = 2pt},
}

\tikzset{
  connection edge/.style ={very thick, gray},
  generaltensor/.style = {    
    rectangle,
    rounded corners = 0.3mm,
    text = white,
    fill = tensorcol,
    draw
    },
  tensorbox/.style={
    generaltensor,
    inner sep = 2pt, 
    minimum height = 1.1\baselineskip,
    minimum width = 1.1\baselineskip
    },
  tensorleg/.style={
    very thick,black!80
    },
  Vertex/.style={
    circle,inner sep=0pt,minimum size=3mm,
    },
  site/.style = {
    minimum width = 1.3em, 
    minimum height = 0.5\baselineskip,
    rounded corners=0.3mm,
    thick,draw=black!40,
    top color=white,bottom color=black!20
    },
  channel/.style = {
    generaltensor,
    minimum height=0.5\baselineskip,
    minimum width=1.5cm
  }, 
  ball/.style ={
      circle,
      radius=0.01cm,
      shading=ball, 
      ball color=black}
  }
  
\newcount\xm
\newcount\ym
\newcount\xx
\newcount\yy

\pgfdeclarelayer{shadow}
\pgfdeclarelayer{background layer}
\pgfdeclarelayer{foreground layer}
\pgfsetlayers{shadow,main}

\colorlet{BboxCol}{gray!15}
\definecolor{LightYellow}{RGB}{255,255,120}
\colorlet{SboxCol}{LightYellow!80!LightGray}
\colorlet{Green}{OliveGreen}
\colorlet{TextGreen}{OliveGreen!80}
\colorlet{TextBlue}{blue!60}
\colorlet{VertexCol1}{blue!30!LightGray}
\colorlet{VertexCol}{black}

\def \ny{8}

\tikzset{overlay re oben/.style = {anchor = north east, yshift = -1cm},
	 vertex/.style = {circle,inner sep = 1.2ex},
	 beschriftung/.style = {transform shape = false},
	 box/.style = {rounded corners =.5ex},
	 edges/.style = {thick, line cap = round, gray},
	 special edges/.style = {edges, ultra thick},
	 red edges/.style = {special edges, red!80},
	 green edges/.style = {special edges, Green},
	 lattice/.style = {scale = .54,transform shape, anchor = center}
	 }  

\tikzset{
	ol/.style = {remember picture, overlay},
	eq-pic/.style = {inner sep = .5pt,draw, very thick, gray, rounded corners = 2pt},
}

\tikzset{
  edges/.style = {thick, line cap = round, gray},
  connection edge/.style ={very thick, gray},
  vertex/.style={
    circle,inner sep=0pt,minimum size=2mm,
    thick,draw=gray,
    fill = lightgray
    },
  tensor1/.style = {    
    generaltensor,
    inner sep = 2pt, 
    minimum height = 1.2\baselineskip,
    minimum width = 1.2\baselineskip,
    fill = tensorcol1,
    draw =tensorcol1border,
    thick
    },
  semicircle/.style = {    
    semicircle,
    rounded corners = 0.3mm,
    text = white,
    fill = tensorcol,
    draw
    },
1qubit/.style={
    generaltensor,
    inner sep = 2pt, 
    minimum height = 1.2\baselineskip,
    minimum width = 1.2\baselineskip,
    fill = tensorcol2,
    draw =tensorcol2border,
    thick
    },
2qubits/.style={
    generaltensor,
    inner sep = 2pt, 
    minimum height = 2.6\baselineskip,
    minimum width = 1.2\baselineskip,
    shading = axis, 
    draw =tensorcol2border,
    thick
    }
}

\newcommand{\measurement}{
\begin{tikzpicture}[scale = .3]
\draw [black!80,very thick,line cap=round,domain=40:140] plot ({.8 * cos(\x)}, {.8 * sin(\x)}); 
\draw [edges,red] (0,0) -- (.5,1);
\end{tikzpicture}
}

\newcommand{\ampelgruen}{
\begin{tikzpicture}[scale = .3,anchor = center]

\node (green) at (0,-1) {};
\node (red) at (0,1) {};

\node [fit = (green) (red),tensor1, draw = DarkGray, fill = MidGray, inner sep = 5pt] {};

\node at (green) [circle,draw = DarkGray,fill = applegreen, thick]{};

\node at (red) [circle,draw = DarkGray,fill = DarkGray, thick]{};

\end{tikzpicture}
}

\newcommand{\ampelrot}{
\begin{tikzpicture}[scale = .3,anchor = center]

\node (green) at (0,-1) {};
\node (red) at (0,1) {};

\node [fit = (green) (red),tensor1, draw = DarkGray, fill = MidGray, inner sep = 5pt] {};

\node at (green) [circle,draw = DarkGray,fill = DarkGray, thick]{};

\node at (red) [circle,draw = DarkGray,fill = DarkRed, thick]{};

\end{tikzpicture}
}

\tikzset{
measure/.style={
    generaltensor,
    inner sep = 2pt, 
    minimum height = 1.2\baselineskip,
    minimum width = 1.2\baselineskip,
    fill = lightgray,
    draw =gray ,
    thick
    },
  tensorbox/.style={
    generaltensor,
    inner sep = 2pt, 
    minimum height = 1.2\baselineskip,
    minimum width = 1.2\baselineskip,
    fill = tensorcol2,
    draw = tensorcol2border
    },
  tensorleg/.style={
    very thick,black!80
    },
  site/.style = {
    minimum width = 1.3em, 
    minimum height = 0.5\baselineskip,
    rounded corners=0.3mm,
    thick,draw=black!40,
    top color=white,bottom color=black!20
    },
  channel/.style = {
    generaltensor,
    minimum height=0.5\baselineskip,
    minimum width=1.5cm
  }, 
  ball/.style ={
      circle,
      radius=0.01cm,
      shading=ball, 
      ball color=black}
  }
  
\PassOptionsToPackage{pdftex,hyperfootnotes=false,pdfpagelabels}{hyperref}
\usepackage{hyperref}
\hypersetup{%
    colorlinks=true, linktocpage=true, pdfstartpage=3, pdfstartview=FitV,%
    breaklinks=true, pdfpagemode=UseNone, pageanchor=true, pdfpagemode=UseOutlines,%
    plainpages=false, bookmarksnumbered, bookmarksopen=true, bookmarksopenlevel=1,%
    hypertexnames=true, pdfhighlight=/O,
    urlcolor=blue, linkcolor=blue, citecolor=green!60!black, 
}

\newtheorem{theorem}{Theorem}
\newtheorem{definition}[theorem]{Definition}
\newtheorem{lemma}[theorem]{Lemma}

\DeclareMathOperator{\Perm}{Perm}
\DeclareMathOperator{\Eb}{\mb E}

\DeclareMathOperator{\Erfc}{\mathrm{Erfc}}

\DeclareMathOperator{\Sym}{Sym}

\newcommand{\ket}[1]{\vert{#1}\rangle}
\newcommand{\bra}[1]{\langle{#1}\vert}

\newcommand{\norm}[1]{\Vert #1 \Vert}

\newcommand{\e}{\mathrm{e}}
\newcommand{\mc}{\mathcal}
\newcommand{\mb}{\mathbb}
\newcommand{\mr}{\mathrm}
\newcommand{\Hmin}{H_{\infty}}

\newcommand{\pbos}{P_{\mathrm{bs},U}}
\newcommand{\piqp}{P_{U_W}}

\newcommand{\dd}{\ \mathrm{d}}

\definecolor{martin}{rgb}{0,.4,1}

\definecolor{jens}{rgb}{0,.8,.5}

\definecolor{christian}{rgb}{.7,.1,0}

\definecolor{dominik}{RGB}{237,16,118}

\newcommand{\myabstract}[1]{
\begin{abstract}
#1
\end{abstract}
\maketitle
}

\makeatletter 
\hypersetup{pdftitle = {Sample complexity of device-independently certified "quantum supremacy"},
	     pdfauthor = {Dominik Hangleiter, Martin Kliesch, Jens Eisert, Christian Gogolin},
	     pdfsubject = {Quantum computation, quantum information},
	     pdfkeywords = {%
	     	Quantum speed-up,
	     	sample complexity, polynomial, exponential,
	     	quantum simulation,
	     	Boson Sampling, sampling,
	     	IQP circuit, instantaneous quantum computation,
        Fourier Sampling,
	     	certification, verification, validation, test, 
        lower bound, 
	     	min-entropy, entropy, flat distribution,
        identity testing,
	    	}
	    }
\makeatother

\newcommand{\hhu}{%
	Institute for Theoretical Physics,
	Heinrich Heine University D{\"u}sseldorf,
	40225 D{\"u}sseldorf,
	Germany
}

\newcommand{\fu}{%
	Dahlem Center for Complex Quantum Systems,
	Freie Universit{\"a}t Berlin,
	14195 Berlin,
	Germany
}

\newcommand{\fumath}{%
  Department of Mathematics and Computer Science, 
  Freie Universit{\"a}t Berlin,
  14195 Berlin, 
  Germany
}

\newcommand{\icfo}{%
	ICFO-Institut de Ciencies Fotoniques,
	The Barcelona Institute of Science and Technology,
	08860 Castelldefels (Barcelona),
	Spain
}

\newcommand{\cologne}{%
 Institute for Theoretical Physics, University of Cologne, 50937 K\"oln, Germany
}

\newcommand{\detailsandproof}{the Supplementary Material~\cite{suppmaterial}}
\newcommand{\refdetailsandproof}{\detailsandproof}

\begin{document}
\ifjournal
\begin{bibunit}
\fi

\title{Sample complexity of device-independently certified ``quantum supremacy''}

\author{Dominik Hangleiter}
\affiliation{\fu}
\author{Martin Kliesch}
\affiliation{\hhu}
\author{Jens Eisert}
\affiliation{\fu}
\affiliation{\fumath}
\author{Christian Gogolin}
\affiliation{\icfo}
\affiliation{\cologne}
\affiliation{Xanadu, 372 Richmond St W, Toronto, M5V 1X6, Canada}

\myabstract{
Results on the hardness of approximate sampling are seen as important stepping stones towards a convincing demonstration of the superior computational power of quantum devices. 
The most prominent suggestions for such experiments include boson sampling, IQP circuit sampling, and universal random circuit sampling. 
A key challenge for any such demonstration is to certify the correct implementation. 
For all these examples, and in fact for all sufficiently flat distributions, 
we show that any non-interactive certification from classical samples and a description of the target distribution requires exponentially many uses of the device. 
Our proofs rely on the same property that is a central ingredient for the approximate hardness results: 
namely, that the sampling distributions, as random variables depending on the random unitaries defining the problem instances, have small second moments.
  }

\ifarxiv\section{Introduction}\fi 
Quantum sampling devices have been hailed as promising candidates for the demonstration of ``quantum (computational) supremacy''\footnote{
Acknowledging the recent debate, we use the term ``quantum (computational) supremacy'' strictly in its established technical meaning~\cite{preskill2013quantum}.
} \cite{preskill2013quantum}.
The goal of any such experiment is to unambiguously demonstrate that quantum devices can solve some tasks both faster and with a more favourable scaling of the computational effort than any classical machine.
At the same time, in the near term it is bound to use those small and computationally restricted quantum devices that are available before the arrival of universal, scalable, and fault-tolerant quantum computers.
This challenge has sparked a flurry of experimental activity \cite{Spring2013,Tillmann2013,Broome2013,Crespi2013,Carolan2013,Spagnolo2013} and prompted the development of better classical sampling schemes for exact \cite{clifford_classical_2017,neville_classical_2017} and imperfect realizations \cite{gogolin_boson-sampling_2013,bremner_achieving_2017,oszmaniec_classical_2018,renema_quantum--classical_2018}.
Due to the reality of experimental imperfections, the key theoretical challenge --- achieved in Refs.~\cite{aaronson_computational_2010,
boixo_characterizing_2016,
bremner_average-case_2016,
bouland_quantum_2018,
gao_quantum_2017,
Supremacy,
hangleiter_anticoncentration_2018,
bouland_complexity_2018,
yoganathan_quantum_2018} using Stockmeyer's approximate counting algorithm \cite{Stockmeyer85ApproxiationSharpP} --- is to prove that even \emph{approximately} sampling from the output distribution of the quantum device is classically hard.

In any such demonstration, the issue of certification is of outstanding importance \cite{shepherd_temporally_2009,gogolin_boson-sampling_2013,aaronson_bosonsampling_2013,aolita_reliable_2015,Hangleiter,bouland_quantum_2018}:
To demonstrate something non-trivial, one not only needs to build a device that is designed to sample approximately from a classically hard distribution but at the same time, one needs to ensure from a feasible number of uses of the device (or its parts) that it actually achieves the targeted task.
How can one convince a skeptical certifier that a quantum device, which supposedly does something no classical machine can do, actually samples from a distribution that is close enough to the ideal target distribution?

The arguably most elegant and most convincing certification would be one based on purely classical data, ideally only the samples produced by the device and a description of the target distribution.
Such certification would be free of additional complexity-theoretic assumptions and device-independent, in that it would be agnostic to all implementation details of the device and would directly certify that the classically defined sampling problem was solved.


\begin{figure}[b]
\begin{tikzpicture}[scale = .65,anchor = center]

\def \l{4};
\def \n{5}

\foreach \x in {1,...,\n}{
    \node (v\x) at (-.2, \x) {$\ket 0$};
    \draw[edges] (v\x) -- (.5, \x);
    \draw[edges] (v\x) -- (\l-.4, \x);

    \draw[edges,DarkGray] (\l,\x+.05) -- (\l+.9, \x + .05);
    \draw[edges,DarkGray] (\l,\x-.05) -- (\l+.9, \x - .05);
    \node[measure] (m\x) at (\l, \x) {\measurement};
    \node (n\x) at (\l+1.35, \x) {\small $S_{i,\x}$};

}

\node (sampler) [fit = (v1) (m\n),inner sep = 1mm, tensor1,
  fill= DarkGray,
    fill opacity = 0.15, draw = DarkGray] {};

\node (lu) at (1,5) {};
\node (bu) at (2.75,1) {};

\node [fit = (lu) (bu),2qubits,
        left color=metalblue!30, right color=metalblue,
        shading angle=60,
        inner sep = .3\baselineskip] { \Large $U$ };

\path (\l,3) ++ (2.5, 0) node  (samples) [anchor = west,draw = DarkGray,thick, rounded corners = 2pt ] {\small $S$} ;

\path (samples.east) ++ (1.5, 0) node (test) [tensor1, shading = axis,
    left color=violet!30, right color=violet,
    shading angle=60,, draw = violet2,minimum size = 3\baselineskip] {$\mathcal{T}_5$};

\path (test.east) ++ (2, -2) node (ampelgruen) {\ampelgruen};
\path (test.east) ++ (2, 2) node (ampelrot) {\ampelrot};

\draw[->, edges] (n3.east) to (samples.west);
\draw[->, edges] (samples.east) to (test.west);

\draw [->, edges] (test.east) to[out=0,in=180] (ampelgruen.west) node [below left,sloped,text = black] {$Q = P_U$};

\draw [-> , edges] (test.east) to[out=0,in=180] (ampelrot.west)  node [above left,sloped,text = black] {$\norm{Q - P_U} > \epsilon$};;

\path (samples.south) ++(0,-1.5) node [draw = DarkGray,thick, rounded corners = 2pt ] (epsilon) {$\epsilon$};

\draw [-> , edges,bend left] (epsilon.north) to (test.south west);

\path (samples.north) ++(0,1.5) node [draw = DarkGray,thick, rounded corners = 2pt ] (dist) {$P_U$};

\draw [-> , edges,bend right] (dist.south) to (test.north west);

\end{tikzpicture}
\caption{\label{fig:testing scheme}%
We consider the problem of certifying probability distributions of the form $P_U(S) = |\bra S U \ket {S_0}|^2$ with an input state $\ket{S_0} = \ket 0 ^{\otimes n}$ and a unitary $U \sim \mu_n$ drawn from some  measure $\mu_n$.
Given $\epsilon > 0$ and access to an arbitrary-precision description of the target distribution $P_U$, the test $\mc T_n$ treats the sampler as a black box and receives a sequence $\mc S = (S_i)_{i = 1}^s \sim Q$ of $s$ samples from an unknown distribution $Q$.
Given $S$ the test is asked to output ``Accept'' if $Q = P_U$ and ``Reject'' if $\norm{Q- P_U}_1 > \epsilon$ with high probability.
}
\end{figure}
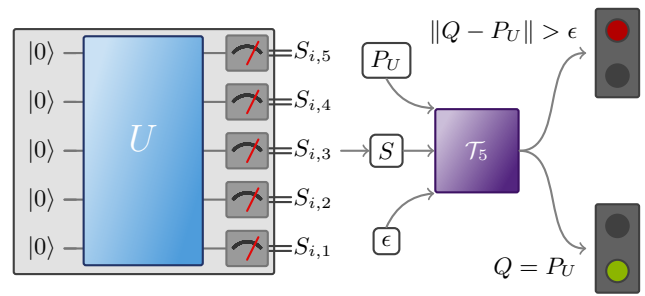

In this work, we rigorously prove for a broad range of sampling problems, specifically for boson sampling  \cite{aaronson_computational_2010}, universal random circuit sampling \cite{boixo_characterizing_2016,bouland_quantum_2018}, IQP circuit sampling \cite{shepherd_temporally_2009,bremner_average-case_2016}, and sampling from post-selected-universal 2-designs \cite{hangleiter_anticoncentration_2018,nakata_generating_2014,Nakata,
yoganathan_quantum_2018,bouland_complexity_2018} that they cannot be efficiently certified from classical samples and a description of the target probability distribution.
Ironically, it turns out that the same property of a distribution that allows to prove the known approximate-hardness results also forbids their non-interactive sample-efficient device independent certification, to the effect that with the known proof methods both properties cannot be achieved simultaneously in such schemes.
We directly bound the sample complexity of certification, which means that we automatically also lower bound the computational complexity and that our results cannot be circumvented by increasing the classical computational power of the certifier\ifarxiv\footnote{This makes our results conceptually different from the observation of Brand\~{a}o.
This observation is based on a result by Trevisan, Tulsiani, and Vadhan \cite{Trevisan2010}, was reported by Aaronson and Arkhipov~\cite{aaronson_bosonsampling_2013} and shows the following:
For most unitaries $U$ drawn from the Haar measure, and any fixed circuit size $T$, there exists a classical ``cheating'' circuit of size polynomially larger than $T$, whose output distribution can not be distinguished from the corresponding boson sampling  distribution by any ``distinguisher'' circuit of size $T$.}\fi.

The specific question of certification we focus on here is (see Figure~\ref{fig:testing scheme}):
Given unlimited computational power and a full
description of the target distribution, how many samples from an unknown distribution are required to guarantee that this distribution is either identical to the target distribution or at least some preset distance away from it?
This problem of distinguishing one (target) distribution from all sufficiently different alternatives is known as \emph{identity testing} \cite{goldreich_introduction_2017} in the property testing literature.
Identity testing is an easier task than its robust version in which the certifier is moreover required to accept a constant-size region around the target distribution \cite{valiant_clt_2010,aolita_reliable_2015}.
At the same time, it is much harder than mere \emph{state-discrimination}, where the task is to differentiate between two fixed distributions.

Lower bounds on the sample complexity of restricted state-discrimination scenarios~\cite{gogolin_boson-sampling_2013} prompted the development of schemes~\cite{aaronson_bosonsampling_2013} that allow to corroborate and build trust in experiments~\cite{Spagnolo2013,Carolan2013,walschaers_statistical_2016}.
This helped spark interest in the problem of device-independent certification --- on which there had not been much progress since~\cite{shepherd_temporally_2009}.
In contrast to previous work~\cite{gogolin_boson-sampling_2013}, here, the certifier is given a full description of the target distribution\footnote{
In particular, the certifier is given the value of all target probabilities to arbitrary precision.
} 
and unlimited computational power.


Our proof\ifjournal, detailed in \refdetailsandproof, \fi makes use of a key property for the proof of hardness of approximate sampling, namely an upper bound on the second moments of the output probabilities with respect to the choice of a random unitary specifying the instance of the sampling problem.
The bound on the second moments implies that the probabilities are concentrated around the uniform distribution and hence an anti-concentration property\ifjournal\footnote{See \refdetailsandproof, Sec.~\ref{sec:second moments}}\fi.
This anti-concentration allows lifting results on the hardness of approximate sampling up to relative errors to ones for additive errors --- provided relative-error approximation of the output probabilities is hard \emph{on average}.
It is thus a key property to prove hardness in the physically relevant case of approximate sampling that prevents a purely classical non-interactive certification of the output distribution, see Figure~\ref{fig:proof overview}.

\begin{figure}
\begin{tikzpicture}[node distance=3cm,align=center,
    every node/.style={
    draw =metalblue2,
    thick,
    rounded corners,text width=2cm,
              shading = axis,
              left color=metalblue!20, right color=metalblue!70,
              shading angle=60}]
    \def\u{.75};
    \def\uu{1.5};
    \node (smb) {Second moment bound};

    \node (ac) [right of=smb,yshift = \u cm] {Anti-concentration};
    \node (ach) at (ac) [yshift = \uu cm] {Average-case hardness};

    \node (has) [right of=ach]{Hardness of approximate sampling};

    \node (sa) at (ach) [yshift=\uu cm]{Stockmeyer's algorithm};

    \node (hes) [right of=sa]{Hardness of exact sampling};

    \node (hme) [right of=smb,yshift=-\u cm,shading = axis,left color=violet!20, right color=violet!60,shading angle=60] {High\\ min-entropy};
    \node (hoc) [right of=hme,shading = axis,left color=violet!20, right color=violet!60,shading angle=60] {Hardness of classical certification};

    \draw[->,edges] (smb) to[out=0,in=180] (ac);
    \draw[->,edges] (smb) to[out=0,in=180] (hme);
    \draw[->,edges] (hme) to[out=0,in=180] (hoc);
    \draw[->,edges] (ach) to[out=0,in=180] (has);
    \draw[->,edges] (ac) to[out=0,in=180] (has);
    \draw[->,edges] (sa) to[out=0,in=180] (hes);
    \draw[->,edges] (sa) to[out=0,in=180] (has);

  \end{tikzpicture}
  \caption{A high level overview of the approximate sampling ``quantum supremacy'' proofs of Refs.~\cite{aaronson_computational_2010,bremner_average-case_2016,boixo_characterizing_2016,miller_quantum_2017,gao_quantum_2017,Supremacy,hangleiter_anticoncentration_2018} using Stockmeyer's algorithm \cite{Stockmeyer85ApproxiationSharpP}.
  Invoking a worst-case hardness result for the calculation of the output probabilities of some circuit family, Stockmeyer's algorithm can be used to prove the hardness of exact sampling.
  The key properties of the output probabilities that allows to prove hardness of \emph{approximate} sampling are that computing these probabilities is even hard \emph{on average} and that the distribution anti-concentrates.
  We show that the same property that is essential to arrive at a hardness result for approximate sampling via anti-concentration also makes it hard to certify from classical samples and a complete description of the target distribution, even with unbounded \emph{computational power}.
  \label{fig:proof overview}
  }
\end{figure}
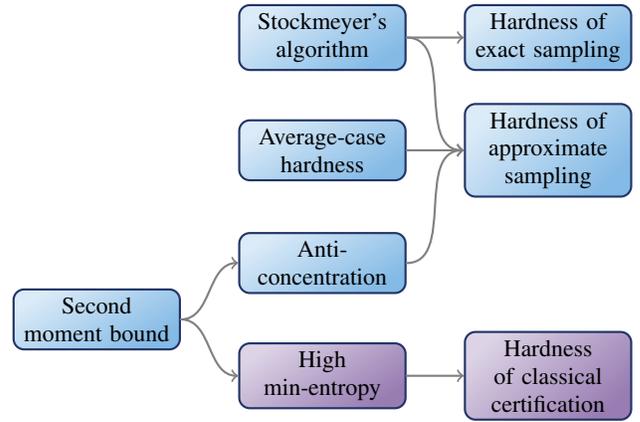

A central ingredient to our proof is a recent result by \citet{valiant_automatic_2017} specifying the optimal sample complexity of certifying a known target distribution $P$.
It can be stated as follows.
Fix a preset distance $\epsilon> 0 $ up to which we want to certify.
Now, suppose we receive samples from a device that samples from an unknown probability distribution $Q$.
Then --- for some constants $c_1,c_2$ --- it requires at least
\begin{equation}
\label{eq:main_valiant_sample_complexity}
c_1 \cdot  \max \left\{\frac1 \epsilon, \frac1{\epsilon^2} \norm{P_{-2\epsilon}^{-\max}}_{2/3} \right\} \,
\end{equation}
and at most
\begin{equation}
\label{eq:main_valiant_sample_complexity_upper}
c_2 \cdot  \max \left\{\frac1 \epsilon, \frac1{\epsilon^2} \norm{P_{-\epsilon/16}^{-\max}}_{2/3} \right\} \,
\end{equation}
many samples to distinguish the case $P = Q$ from the case $\norm{P-Q}_1\geq \epsilon$ with high probability.
Here $\norm{\cdot}_1$ denotes the $\ell_1$-norm reflecting the total-variation distance.
The central quantity determining the sample complexity of certification is thus the quasi-norm $\norm{P_{-\epsilon}^{-\max}}_{2/3}$ which is defined as follows.
First, find the truncated distribution $P_{-\epsilon}^{-\max}$ by removing the tail of the target distribution $P$ with weight at most $\epsilon$ as well as its largest entry, see Figure~\ref{fig:p-epsilon-max}.
Then, take the $\ell_{2/3}$-norm as given by $\norm{x}_{2/3} = ( \sum_i |x_i|^{2/3} )^{3/2}$ for a vector $x$ with entries $x_i$.

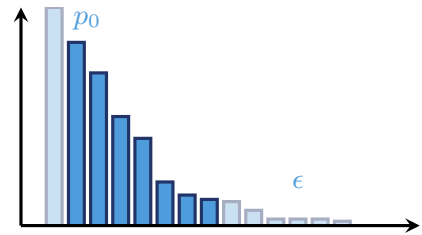
\begin{figure}[b]

\begin{tikzpicture}[scale = 1.2,anchor = center]
    \begin{axis}[
      ybar,
      axis on top,
      bar width=5pt,
      ymin=0,
      ymax =50,
      xmin = -1,
      xmax = 14,
      ytick=\empty,
      xtick=\empty,
      axis x line=bottom,
      axis y line=left,
      enlarge x limits=0.1,
      line width = 1pt,
      width = 6cm,
      height = 4cm
    ]

      \addplot[fill=metalblue!30,draw = metalblue2!40] coordinates {
        (0,50)
      };
        \addplot[fill=metalblue,draw = metalblue2] coordinates {
        (0,42)
        (1,35)
        (2,25)
        (3,20)
        (4,10)
        (5,7)
        (6,6)
      };
      \addplot[fill=metalblue!30,draw = metalblue2!40] coordinates {
        (6,5.5)
        (7,3.5)
        (8,1.5)
        (9,1.5)
        (10,1.5)
        (11,1)
      };

      \node at (axis cs:10,10) {\color{metalblue}$\epsilon$};

      \node at (axis cs:0.5,47) {\color{metalblue}$p_{0}$};
    \end{axis}

  \end{tikzpicture}
  \caption{The vector $P_{-\epsilon}^{-\max}$ is obtained from $P$ by removing the largest element $p_0$ of $P$ as well as the smallest probabilities that accumulate to a total weight bounded by $\epsilon$.}
  \label{fig:p-epsilon-max}
\end{figure}

We now proceed in two steps.
First, we show lower and upper bounds on the quantity $\norm{P_{-\epsilon}^{-\max}}_{2/3}$ in terms of the largest probability $p_0$ occurring in $P$ and its support $\norm{P_{-\epsilon}^{-\max}}_0$ as given by
\begin{equation}
  \label{eq:2/3_bounds_maintext}
\begin{split}
  p_0^{-\frac12}  & \left(1- \epsilon -p_0 \right)^{3/2}
  \leq \norm{P_{- \epsilon}^{- \max} } _{2/3}
  \\
  & \leq \left (1- p_0 \right) \norm{P_{- \epsilon}^{-\max} }_0^{\frac12} .
\end{split}
\end{equation}
Then it follows from Eqs.~\eqref{eq:main_valiant_sample_complexity} and \eqref{eq:2/3_bounds_maintext} that the sample complexity of certifying a distribution $P$ up to a constant total-variation distance $\epsilon$ is essentially lower bounded by $1/\sqrt{p_0}$.
Hence, if $P$ is exponentially flat in the sense that the largest probability is exponentially small in the problem size (here, the number of particles), $\epsilon$-certification requires exponentially many samples.
Conversely, if $P_{- \epsilon/16}^{-\max}$ is supported on polynomially many outcomes only, sample-efficient certification is possible by the converse bound \eqref{eq:main_valiant_sample_complexity_upper}.

Second, we connect this result to the output distributions of ``quantum supremacy'' schemes.
In all schemes that rely on the Stockmeyer argument, the problem instances are defined in terms of a unitary \ifjournal $U$ \fi  that is \emph{randomly chosen} from some restricted family, e.g., linear optical circuits in the case of boson sampling \cite{aaronson_computational_2010} or random universal circuits \cite{boixo_characterizing_2016,bouland_quantum_2018} in a qubit architecture \ifjournal , captured by a respective measure $\mu_n$ on the $n$-particle unitary group\fi.
\ifarxiv
Specifically, we prove that with high probability over the choice of the random unitary, the distribution over outputs associated with this unitary is exponentially flat.
\fi
\ifjournal 
Specifically, we prove that with high probability over the choice of the random unitary, the distribution over outputs $P_U $ associated with this unitary (induced via $P_U(S) = | \bra S U \ket{S_0} |^2$, $U \sim \mu_n$) will be exponentially flat in the sense that $\Hmin(P_U) \geq \Omega(n)$. 
Here, $\Hmin(P) = - \log \max_x p_x$ is the min-entropy of $P$.
We show that, ironically, this follows from an upper bound on the second moments of the output probabilities, which is at the same time a central ingredient in the Stockmeyer hardness argument for \emph{approximate} sampling.
Specifically, we prove that with probability at least $1- \delta$ over the choice of $U$ (see Lemma~\ref{lem:anticon-minentropy} in \refdetailsandproof)
\begin{align}
  \Hmin(P_U) \geq \frac12 \left( \log  \delta - \log \sum_{S } \mathbb{E}_{U\sim \mu_n}[P_U(S)^2] \right).
\end{align}
\fi

Putting everything together we obtain lower bounds on the sample complexity of certification for boson sampling, IQP circuit sampling and random universal circuit sampling with (sufficiently many) $n$ particles.
In all of these cases, the sample complexity scales at least as fast as
\begin{equation}
  \frac 1 {\epsilon^2} (2^n \delta)^{1/4} \, ,
\end{equation}
with probability at least $1-\delta$ over the random choice of the unitary.

The upshot is: a key ingredient of the proof of approximate sampling hardness as effected by the random choice of the unitary prohibits sample-efficient certification.
%

%
We show that one cannot hope for purely classical, non-interactive, device-independent certification of the proposed quantum sampling problems.
This highlights the importance of devising alternative schemes of certification, or improved hardness results for more peaked distributions.
We hope to stimulate research in such directions.

A particularly promising avenue of this type of certification has been pioneered by \citet{shepherd_temporally_2009}:
By allowing the certifier to choose the classical input to the sampling device rather than drawing it fully at random, it is under some plausible cryptographic assumptions possible to efficiently certify the correct implementation of a quantum sampler from its classical outcomes.
This is facilitated by checking a previously hidden bias in the obtained samples and has been achieved for a certain family of IQP circuits \cite{shepherd_temporally_2009}.
However, in contrast to Ref.~\cite{bremner_average-case_2016}, there is no approximate sampling hardness result for this family. 

Focusing on so-called \emph{relational problems} as opposed to sampling problems, it has been argued via new complexity-theoretic conjectures that the task \emph{HOG} of outputting the \emph{heavy outcomes} of a quantum circuit (those outcomes with probability weight larger than the median of its output distribution) is classically intractable \cite{aaronson_complexity-theoretic_2017}.
Clearly, this task is sample-efficiently checkable via its in-built bias, but still requires exponential classical computation to determine the probabilities of the obtained samples, which are compared to the median.

Taking a pragmatic stance, one can make additional assumptions on the device.
In fact, only recently has it been shown \cite{bouland_quantum_2018} that cross-entropy measures \cite{boixo_characterizing_2016} provide direct bounds on the total-variation distance provided the entropy of the real distribution is larger than that of the target distribution.
One might also be content with weaker notions of certification in total-variation distance such as the certification of a coarse-grained version of the full output distribution \cite{wang_certification_2016}. 
Coarse-graining procedures are practically useful as corroboration schemes when distinguishing against plausible alternative distributions such as the uniform distribution, but of course fail to certify against adversarial distributions on the full sample space.
All such approaches yield sample-efficient certificates that require exponential computational effort, rendering them feasible at least for intermediate-scale devices.

Another way to certify a sampling device is the certification of the entire machine from its components.
However, such schemes need to make assumptions about the absence of unwanted influences between the components such as crosstalk.
In a similar vein, one can make use of implementation details and give the certifier some quantum capabilities such as access to a small quantum computer \cite{wiebe_hamiltonian_2014}, the ability to manipulate single qubits \cite{Mills2017}, or to measure the output quantum state in different bases with trusted quantum detectors \cite{Hangleiter,Badescu2017} to devise certificates even in non-iid.\ settings \cite{takeuchi_verification_2018}.
In this way, one can partially trade-in the simplicity of sampling schemes for better certifiability.

%

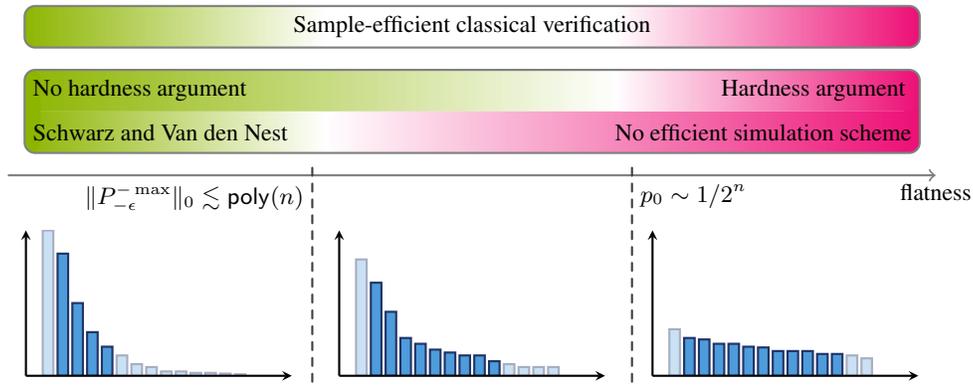
\begin{figure*}[t]
\begin{tikzpicture}[scale = .85,anchor = center]

  \def\L{7.25};
  \def\l{7};
  \def\h{2};
  \def\H{3};
  \def\ny{1.};
  \def\hy{.65}
  \def\pol{-2.5};
  \def\ex{+2.5};



  \draw [->,edges] (-\L,0) -- node [at end, black, below, align = left] {flatness}
  (\L,0);

  \node (lt) at (-\l,\ny + \hy) {};
  \node (lt2) at (-\l,\ny) {};
  \node (rt) at (\l,\ny - \hy) {};
  \node (lb) at (-\l , 2 * \ny + \hy -\hy) {};
  \node (rb) at (\l ,2 * \ny + \hy) {};

 \node (rbtc) at (\ex -.25,2 * \ny + \hy  -\hy) {};

  \node (lbtc) at (\pol + .25,2 * \ny + \hy) {};
  \node (ltbc) at (\pol +.25,\ny -\hy) {};
  \node (ltbc2) at (\pol +.25 + 0.66*\l,\ny) {};
  \node (rtc) at (\ex - .5, - 2, 0) {};

  \draw [rounded corners, draw = none,shading angle = 90, left color = applegreen, right color = white] (lb) rectangle  (lbtc) ;

  \draw [rounded corners, draw = none,shading angle = 90, left color = applegreen, right color = white] (lt) rectangle  (ltbc2) ;
  
  \draw [rounded corners, draw = none,shading angle = 90, left color = applegreen, right color = white] (lt2) rectangle  (ltbc) ;


  \draw [rounded corners, draw = none, fill =pink
  ] (\l,\ny-\hy) rectangle   (\l - .25,\ny + \hy) ;

  \draw [rounded corners, draw = none, fill =applegreen
  ] (-\l,\ny-\hy) rectangle   (-\l + .25,\ny + \hy) ;

  \draw [rounded corners, draw = none,shading angle = 90, left color =white, right color =pink] (rt) rectangle   (\pol+.25,\ny ) ;

  \draw [rounded corners, draw = none,
  shading angle = 90, left color =white, middle color = pink, right color =pink
  ] (\l,\ny) rectangle   (\ex - .25,\ny + \hy) ;

 \draw [rounded corners, draw = none,shading angle = 90, left color = white, right color = pink] (rb) rectangle  (rbtc) ;

  \draw [edges, rounded corners, fill = none] (lt) rectangle  (rt) ;

  \draw [edges, rounded corners,fill = none] (lb) rectangle  (rb) ;


  \node at (\l,\ny + \hy/2) [anchor = east] {Hardness argument \quad};
  \node at (-\l,\ny + \hy/2) [anchor = west] {No hardness argument};
  \node at (\l,\ny - \hy/2) [anchor = east] {No efficient simulation scheme};
  
  \node at (-\l,\ny -\hy/2) [anchor = west] {Schwarz and Van den Nest \quad};

  \node at (0,2 * \ny + 0.5*\hy ) {Sample-efficient classical verification};

  \node (poln)  at (\pol, .25) {};
  \node (exn)  at (\ex, .25) {};

  \node [anchor = north east] at (\pol,0) {$\norm{P_{-\epsilon}^{-\max}}_0 \lesssim \poly(n)$};

  \node [anchor = north west] at (\ex,0) {$p_0 \sim 1/2^n$};

  \node (flat) at (-.7 * \l,- \h) {

\begin{tikzpicture}[scale = .8,anchor = center,every node/.style={}]
    \begin{axis}[
      ybar,
      axis on top,
      bar width=5pt,
      ymin=0,
      ymax =50,
      xmin = -1,
      xmax = 14,
      ytick=\empty,
      xtick=\empty,
      axis x line=bottom,
      axis y line=left,
      enlarge x limits=0.1,
      line width = 1pt,
      width = 6cm,
      height = 4cm
    ]

      \addplot[fill=metalblue!30,draw = metalblue2!40] coordinates {
        (0,50)
      };
        \addplot[fill=metalblue,draw = metalblue2] coordinates {
        (0,42)
        (1,25)
        (2,15)
        (3,10)
      };
      \addplot[fill=metalblue!30,draw = metalblue2!40] coordinates {
      (3,7)
         (4,4)
        (5,3)
        (6,1.5)
        (7,1.5)
        (8,1)
        (9,1)
        (10,.75)
        (11,.5)
      };
    \end{axis}

  \end{tikzpicture}

  };

  \node (middle) at (0 ,- \h) {

\begin{tikzpicture}[scale = .8,anchor = center,every node/.style={}]
    \begin{axis}[
      ybar,
      axis on top,
      bar width=5pt,
      ymin=0,
      ymax =50,
      xmin = -1,
      xmax = 14,
      ytick=\empty,
      xtick=\empty,
      axis x line=bottom,
      axis y line=left,
      enlarge x limits=0.1,
      line width = 1pt,
      width = 6cm,
      height = 4cm
    ]

      \addplot[fill=metalblue!30,draw = metalblue2!40] coordinates {
        (0,40)
      };
        \addplot[fill=metalblue,draw = metalblue2] coordinates {
        (0,32)
        (1,22)
        (2,13)
        (3,11)
        (4,9)
        (5,8)
        (6,7)
        (7,7)
        (8,5)
      };
      \addplot[fill=metalblue!30,draw = metalblue2!40] coordinates {

        (8,4)
        (9,3)
        (10,3)
        (11,3)
      };
    \end{axis}

  \end{tikzpicture}

  };
  \node (peaked) at (.7 * \l,- \h) {

\begin{tikzpicture}[scale = .8,anchor = center,every node/.style={}]
    \begin{axis}[
      ybar,
      axis on top,
      bar width=5pt,
      ymin=0,
      ymax =50,
      xmin = -1,
      xmax = 14,
      ytick=\empty,
      xtick=\empty,
      axis x line=bottom,
      axis y line=left,
      enlarge x limits=0.1,
      line width = 1pt,
      width = 6cm,
      height = 4cm
    ]

      \addplot[fill=metalblue!30,draw = metalblue2!40] coordinates {
        (0,16)
      };
        \addplot[fill=metalblue,draw = metalblue2] coordinates {
        (0,13)
        (1,12.5)
        (2,11)
        (3,11)
        (4,10)
        (5,9.75)
        (6,8.5)
        (7,8.5)
        (8,8.5)
        (9,7.5)
        (10,7.5)
      };
      \addplot[fill=metalblue!30,draw = metalblue2!40] coordinates {

        (10,7)
        (11,6)
      };
    \end{axis}

  \end{tikzpicture}

  };

  \draw [dashed,edges,DarkGray] (exn) -- (exn|-flat.south);
  \draw [dashed,edges,DarkGray] (poln) -- (poln|-flat.south);




\end{tikzpicture}
\caption{
\label{fig:roominthemiddle}
Hardness and certification in terms of the flatness of $P_{-\epsilon}^{-\max}$ for the example of IQP circuits \cite{shepherd_temporally_2009,bremner_average-case_2016} on $n$ qubits as obtained from the present result and the classical simulation algorithm of \citet{schwarz_simulating_2013}.
There, it is shown that a certain natural family of quantum circuits (including IQP circuits) can be efficiently simulated on a classical computer if the output distribution is essentially concentrated on a polynomial number of outcomes only.
In this case, i.e., for $\norm{P_{-\epsilon}^{-\max}}_0 \lesssim \poly(n)$, the output distribution is also sample-efficiently certifiable as the bounds~\eqref{eq:main_valiant_sample_complexity_upper} and \eqref{eq:2/3_bounds_maintext}
show.
Their classical simulation algorithm breaks down if the distribution is essentially spread out over more than polynomially many outcomes, and we even have a rigorous hardness argument by \citet{bremner_average-case_2016} for exponentially flat distributions.
Conversely, the number of samples required for certification becomes prohibitively large if the distribution is exponeentially spread out, as measured by the $\ell_{2/3}$-norm \eqref{eq:main_valiant_sample_complexity}.
Nevertheless,
as we illustrate here,
there could be ``room in the middle'' where, for reasonably but not exponentially flat distributions, one may hope to find tasks that are both classically intractable and sample-efficiently certifiable in a device-independent fashion.
}
\end{figure*}

It is interesting to note the connection of our result with results on classical simulation.
Similarly to our findings for the case of certification, \citet{schwarz_simulating_2013} find that for certain natural families of quantum circuits (including IQP circuits) classical simulation is possible for highly concentrated distributions, but impossible for flat ones, see Figure~\ref{fig:roominthemiddle}.
This again gives substance to the interesting connection between superior computational power, the flatness of the distribution and the impossibility of an efficient certification.

Curiously, at the same time, the property that prohibits sample-efficient certification is by no means due to the hardness of the distribution.
It is merely the flatness of the distribution on an exponential-size sample space as effected by the random choice of the unitary that is required for the approximate hardness argument via Stockmeyer's algorithm and standard conjectures.
The uniform distribution on an exponentially large sample space, which is classically efficiently samplable, can also not be sample-efficiently certified.

A further noteworthy connection is that to Shor's algorithm.
The output distribution of the quantum part of Shor's algorithm is typically spread out over super-polynomially many outcomes and can hence neither be efficiently simulated via the algorithm of \citet{schwarz_simulating_2013}, nor certified as we show here.
However, after the classical post-processing, the output distribution is strongly concentrated on few outcomes --- the factors --- from which one can verify the correct working of the algorithm.
A certification of the intermediate distribution is simply not necessary to demonstrate a quantum speedup in Shor's algorithm, as its speedup is derived from it solving a problem in \np~and not from it sampling close to a hard distribution.
This shows that while intermediate steps of a computation might not be certifiable, the final outcome may well be.
Whether this is enough to demonstrate a speedup depends on the nature of the hardness argument.
In fact, the abovementioned task HOG \cite{aaronson_complexity-theoretic_2017} bears many similarities to factoring and its certifiability from the outcomes of the algorithm.

We hope that our result will stimulate research into new ways of proving hardness of approximate sampling tasks that are more robust than those based on anti-concentration, as well as into devising alternative verification schemes possibly based on mild and physically reasonable assumptions on the sampling device or the verifier.

\ifjournal


\section*{Acknowledgements} 
We are grateful to Adam Bouland who pointed us to the literature on property testing and thus provided the missing clue for finishing this project.
We would like to thank Ashley Montanaro, Bill Fefferman, Tomoyuki Morimae, Martin Schwarz, and Juan Bermejo-Vega for fruitful discussions and Aram Harrow and Anand Natarajan for sharing an early version of their related work. 
Finally, we would like to thank two anonymous referees for their thorough proof-checking and interesting questions which helped improve the manuscript.

D.\ H.\ and J.\ E.\ acknowledge support from the ERC (TAQ), the Templeton Foundation, the DFG (EI 519/14-1, EI 519/9-1, EI 519/7-1, CRC 183), and the European Union’s Horizon 2020 research and innovation programme under grant agreement No.\ 817482 (PASQUANS).
C.\ G.\ acknowledges support by the European Union’s Marie Sk\l{}odowska-Curie Individual Fellowships
(IF-EF) programme under GA: 700140 as well as financial support from ARO under contract W911NF-14-1-0098 (Quantum Characterization, Verification, and Validation), the Spanish Ministry MINECO (National Plan 15 Grant: FISICATEAMO No.\ FIS2016-79508-P, SEVERO OCHOA No.\ SEV-2015-0522), Fundació Cellex, Generalitat de Catalunya (Grants No.\ SGR 874 and No.\ 875, CERCA Programme, AGAUR Grant No.\ 2017 SGR 1341 and CERCA/Program), ERC (CoG QITBOX and AdG OSYRIS), EU FETPRO QUIC, EU STREP program EQuaM (FP7/2007–2017, Grant No.\ 323714), and the National Science Centre, Poland-Symfonia Grant No.\ 2016/20/W/ST4/00314.

\fi

\ifjournal

\putbib
\end{bibunit}
\begin{bibunit}
  \cleardoublepage
  \widetext
  \setcounter{page}{1}
  \setcounter{equation}{0}
  \setcounter{footnote}{0}
  \thispagestyle{empty}

  \begin{center}
  \textbf{\large Supplementary Material for\\Sample complexity of device-independently certified ``quantum supremacy''}\\
  \vspace{2ex}
  Dominik Hangleiter, Martin Kliesch, Jens Eisert, and Christian Gogolin
  \end{center}
  \twocolumngrid
  \renewcommand\thesection{S\arabic{section}}
\fi

	\section{Setup and definitions}

Let us begin the technical part of this work by setting the notation.
We use the Landau symbols $O$ and $\Omega$ for asymptotic upper and lower bounds and $\Theta$ for their conjunction.
For any $j \in \mb Z^+$ we employ the short hand notation $[j] \coloneqq \{1,\dots,j\}$ for the range.
By $\log$ we denote the logarithm to basis $2$.
For any vector $x \in \mb R^n$ we define
$\norm{x}_{\infty} \coloneqq \max_{i \in [n]} |x_i|$
and
$\norm{x}_{p} \coloneqq \left( \sum_{i=1}^n |x_i|^p \right)^{1/p}$ for
$0 < p < \infty$ and take
$\norm{x}_0 \coloneqq |\{ i \in [n]: x_i \neq 0 \}|$
to denote the number of non-zero elements of $x$.
Thus, $\norm{\cdot}_{p}$ is the standard $\ell_p$-norm whenever $p\geq 1$.
For $p\in(0,1)$, $\norm{\cdot}_{p}$ no longer satisfies the triangle inequality, but it is obviously absolutely homogeneous and still defines a quasinorm.
We will also make use of the $\alpha$-R\'enyi entropies, which for any probability vector $P = (p_1, \ldots, p_n), \, p_i \geq 0 ,\, \sum_i p_i = 1$ and $0 \leq \alpha \leq \infty, \alpha \neq 1$ are defined to be
\begin{align}
  H_{\alpha}(P) \coloneqq \frac\alpha{1-\alpha} \log \norm{P}_{\alpha} .
\end{align}
We refer to $\Hmin(P) = -  \log \max_{i \in [n]} p_i$ as the \emph{min-entropy} of $P$.

We are now in the position
to formalize the notion of a test that certifies that a given device indeed samples from a distribution sufficiently close to a given target distribution.
More precisely, we consider a family of sample spaces $\mc E_n$ with $n \in \mathbb{Z}^+$.
The parameter $n$ will be the natural problem size in the concrete examples below.
The object of interest is a classical algorithm $\mc T_n$ which receives a description of a target distribution $P_n$ over $\mc E_n$, and a sequence $\mc S \sim Q^s$ of $s$ samples $S_1, \ldots, S_s \in \mc E_n$ that have been drawn i.i.d.\ from some distribution $Q$ over $\mc E_n$ and must output $1$ or $0$ for ``accept'' or ``reject'', respectively.
We illustrate this notion of certification in Figure~\ref{fig:testing scheme}.

\begin{definition}[Certification test]
\label{def:certification}
  For any $n$ let $P$ be a (target) probability distribution on a sample space $\mc E$.
  We call
  $\mc T: \mc E^s \to \{0,1\}$
  an $\epsilon$-\emph{certification test} of $P$ from $s$ samples if the following completeness and soundness conditions are satisfied for any distribution $Q$ over $\mc E$:
  \begin{align}
    Q = P
    & \ \Rightarrow \
    \Pr_{ \mc S \sim Q^s } [\mc T(\mc S) =1 ] \geq \frac23,
    \\
    \norm{P-Q}_{1} > \epsilon
    &\ \Rightarrow\
    \Pr_{ \mc S\sim Q^s} [\mc T(\mc S) =1 ] < \frac13 \, .
  \end{align}
  For a family $\{P_n\}$ of probability distributions we call a family of tests $\{\mc T_{n}\}$ a \emph{sample-efficient $\epsilon$-certification test} if for every $n$ $\mc T_{n}$ is an $\epsilon$-certification test from $s \in O(\poly(n,1/\epsilon))$ samples.
\end{definition}

Our notion of certification is \emph{device-independent} in the sense that it does not assume anything about the internal working of the sampler (not even whether it is quantum or classical), but uses only the classical samples it outputs and a classical description of the target distribution.
Among such device-independent certification scenarios, our scenario is the most general one in the sense that the certifier is given \emph{all} the information contained in the target distribution.
In particular, it is crucial that we explicitly allow the certification test $\mc T_n$ to depend on all details of the target distribution $P_n$.

As we are not concerned with the computational complexity of the test, but only its sample complexity, we allow the certification algorithm \emph{unlimited} computational power.
In particular, it does not matter how exactly $\mc T_n$ is given access to a description of $P_n$, but for the sake of concreteness $\mc T_n$ can be thought of as having access to an oracle that provides the probabilities $P_n(S)$ of all $S \in \mc E_n$ up to arbitrary precision.
Sample-efficiency is clearly a necessary requirement for computational efficiency of a test, as any test takes at least the time it needs to read in the required number of samples, so that lower bounds on the sample complexity are stronger than such on the computational complexity.

We note that our notion of certification corresponds to what in the literature on property testing \cite{goldreich_introduction_2017} is called \emph{identity testing} with a fixed target distribution.
It stands in contrast to the previously considered task of \emph{state discrimination} \cite{gogolin_boson-sampling_2013,aaronson_bosonsampling_2013}, where the task is to decide from which of two given distributions $P$ or $Q$ a device samples.
$\epsilon$-certification in the sense of Definition~\ref{def:certification} is more demanding in the sense that $P$ has to be distinguished from \emph{all} distributions $Q$ such that $\norm{P-Q}_{1} \geq \epsilon$.
It is precisely this type of certification that is necessary to convince a skeptic of ``quantum supremacy'' via, say boson sampling, as the hardness results on approximate boson sampling only cover distributions within a small ball in $\ell_1$-norm around the ideal target distribution.
A device sampling from a distribution further away from the ideal distribution, might still be doing something classically intractable, but this cannot be concluded from the hardness of approximate boson sampling.


\section{No certification of flat distributions}

This section is concerned with the question of whether distributions with a high \emph{min-entropy} can be certified in a sample-efficient way.
The main insights into this question come from a work by \citet{valiant_automatic_2017} on property testing, which gives a \emph{sample-optimal} certification test (up to constant factors) for any fixed distribution $P$, as well as a lower bound on the sample complexity of certification.
The result is stated in terms of an $\ell_{2/3}$-norm of a vector obtained from the distribution.
Our main technical contribution is to find bounds on these quasi-norms that are relevant in the context of certifying ``quantum supremacy'' distributions.

To state the main result of Ref.~\cite{valiant_automatic_2017}, we adapt their following notation and illustrate it in Figure~\ref{fig:p-epsilon-max}.
For any vector of non-negative numbers $P$,
\begin{enumerate}[(i)]
\item let $P^{- \mathrm{max}}$ be the vector obtained from $P$ by setting the largest entry to zero, and
\item let $P_{-\epsilon}$ be the vector obtained from $P$ by iteratively setting the smallest entries to zero, while the sum of the removed entries remains upper bounded by $\epsilon>0$.
\end{enumerate}
It turns out that the optimal sample complexity for $\epsilon$-certifying any distribution $P$ is essentially given by $\frac 1 {\epsilon^2} \norm{P_{-\epsilon}^{- \mathrm{max}} } _{2/3}$.
The intuition is that any $\epsilon$ deviation from $P$ that is contained in either the largest probability or the tail of the distribution is easily detected.
Intuitively, this is because a constant deviation in these parts of the distribution will be visible in the samples obtained with high probability \cite{valiant_automatic_2017}.
More precisely, the following upper and lower bounds on the sample complexity of certification hold:
\begin{theorem}[Optimal certification tests \cite{valiant_automatic_2017}]
\label{thm:optimal tests}
  There exist constants $c_1, c_2 > 0 $ such that for any $\epsilon > 0$ and any target distribution $P$, there exists an $\epsilon$-certification test from $c_1\,\max \{ \frac 1 \epsilon, \frac1{\epsilon^2}\,\norm{P_{- \epsilon/16}^{- \mathrm{max}} } _{{2/3}} \}$ many samples, but there exists no $\epsilon$-certification test from fewer than $c_2\, \max \{ \frac 1 \epsilon, \frac 1{\epsilon^2}\,\norm{P_{- 2\epsilon}^{- \mathrm{max}} } _{{2/3}}  \}$ samples.
\end{theorem}
We note that $\norm{P_{- \epsilon}^{- \mathrm{max}} } _{{2/3}} \leq \norm{P} _{{2/3}}$ for any $P$, and in many cases the former is only a constant factor away from the latter.
We obtain the following general bounds on $\norm{P_{- \epsilon}^{- \mathrm{max}} } _{{2/3}}$ in terms of the min-entropy and support of $P$:
\begin{lemma}[Bounds on $\norm{P_{- \epsilon}^{- \mathrm{max}} } _{{2/3}}$]
\label{lem:bounds} \quad
\begin{equation}
\label{eq:2/3 bounds}
\begin{split}
  2^{\frac12\Hmin(P)} & \left(1- \epsilon - 2^{-\Hmin(P)} \right)^{3/2}
  \leq \norm{P_{- \epsilon}^{- \mathrm{max}} } _{{2/3}} \\& \leq \left (1 - 2^{-\Hmin(P)} \right) \norm{P_{- \epsilon}^{-\max} }_0^{\frac12} .
\end{split}
\end{equation}
\end{lemma}
To get a feeling for what these bounds imply, let us consider two special cases and sufficiently small $\epsilon$.
If for some constant $\kappa$ it holds that
$\Hmin(P) = \log(\kappa\,|\mc E_n|)$,
they imply the following lower bound on the required minimal number of samples, $s_{\mr {min}} $:
\begin{align}
  s_{\mr {min}}^2 & \geq c_2^2\kappa\,\frac{|\mc E_n|}{\epsilon^4}\left( 1- 2 \epsilon - \frac{1}{\kappa\,|\mc E_n|}\right)^3.
\end{align}
For all distributions whose min-entropy is essentially given by the logarithm of the size $|\mc E_n|$ of the sample space, the sample complexity for certification thus scales at least as the square root of that size.
If, on the contrary, $P_{-\epsilon/16}$ has support on at most $s \geq \norm{P_{-\epsilon/16}}_0$ many probabilities we have the following upper bound
\begin{equation}
  s_{\mr {suf}} \leq c_1\,\frac{1- \frac \epsilon{16}}{\epsilon^2}\,\sqrt{s}
\end{equation}
on the number of samples $s_{\mr {suf}}$ that is sufficient for $\epsilon$-certification.
This bound implies that distributions supported only on polynomially many outcomes can be certified from polynomially many samples.

\begin{proof}[Proof of Lemma~\ref{lem:bounds}]
For the lower bound, we use that concavity of the function $x \mapsto x^{2/3}$ implies that for any fixed $x^\ast > 0$ and any $0 \leq x \leq x^\ast$ we have
\begin{equation}
  x^{2/3} \geq \frac{{x^\ast}^{2/3}}{x^\ast} \, x = {x^\ast}^{-1/3}\,x
\end{equation}
and thus for any (not necessarily normalized) $\tilde P \coloneqq (\tilde p_1, \dots, \tilde p_{\tilde n})$ with $\tilde p_i \geq 0$
\begin{align}
  \norm{\tilde P}_{{2/3}}^{2/3}&  = \sum_{i=1}^{\tilde n} \tilde p_i^{2/3} \geq  \sum_{i=1}^{\tilde n} (\norm{\tilde P}_\infty^{-1/3}\, \tilde p_i) \\
  &= \norm{\tilde P}_\infty^{-1/3} \,\norm{\tilde P}_1.
\end{align}
Using this for $\tilde P = P_{- \epsilon}^{- \mathrm{max}}$ and that both $\norm{P_{- \epsilon}^{- \mathrm{max}}}_\infty \leq \norm{P}_\infty$ and $\norm{P_{- \epsilon}^{- \mathrm{max}} }_{1} \geq 1 - \epsilon - \norm{P}_\infty$ finally implies the lower bound.

For the upper bound, we use that for any vector $v$
and $0< p < q \leq \infty$ (see, e.g., Ref.~\cite[Eq.~(A.3)]{foucart_mathematical_2013})
\begin{align}
	\norm{v}_p \leq s^{\frac1p - \frac1q} \norm{v}_q ,
\end{align}
where $s \geq \norm{v}_0$.
Inserting $p = 2/3$ and $q =1$, one obtains for $v= P_{- \epsilon}^{- \mathrm{max}}$
\begin{equation}
\norm{P_{- \epsilon}^{- \mathrm{max}} } _{{2/3}}
\leq
\norm{P_{- \epsilon}^{- \mathrm{max}} }_0^{\frac12} \norm{P_{- \epsilon}^{- \mathrm{max}} }_1 \, .
\end{equation}
\end{proof}

Valiant and Valiant's result \cite{valiant_automatic_2017} also has immediate consequences on the certifiability of post-selected probability distributions, such as those arising in boson sampling  \cite{aaronson_computational_2010}.
A certification algorithm has to distinguish the target distribution $P$ from all probability distributions that are at least $\epsilon$-far away from $P$.
That is true, in particular, for distributions that differ from $P$ by at least $\epsilon$ in $\ell_1$-norm only on some part $\mc F$ of the sample space, but are identical with $P$ on its complement $\mc F^c$.
Intuitively one can expect that to distinguish such distributions, samples from $\mc F^c$ do not help.
One might hence expect that it should be possible to lower bound the sample complexity of certifying the full distribution by the sample complexity of the post-selected distribution on some subspace $\mc F$ of the sample space, at least as long as the post-selection probability is not too low.

To make this intuition precise, define for any probability distribution $P$ and any subset $\mc F \subset \mc E$ the restriction $P_{\restriction \mc F} \coloneqq (p_i)_{i \in \mc F}$ of $P$ to $\mc F$ (no longer normalized),
as well as the post-selected probability distribution $P_{\mc F} \coloneqq P_{\restriction \mc F}/P(\mc F)$, with post-selection probability $P(\mc F) \coloneqq \norm{P_{\restriction \mc F}}_{1}$.
\begin{lemma}[Lower bounds with post-selected distributions]
  \label{lem:postselection}
  Let $P$ be a probability distribution on $\mc E$.
  Then with $c_2$ the constant from Theorem~\ref{thm:optimal tests} and for any $\epsilon > 0$ and $\mc F \subset \mc E$, there exists no $\epsilon$-certification test of $P$ from fewer than  $c_2 \max \{\frac 1\epsilon, \frac 1 {\epsilon^2} P(\mc F)\,  \norm{(P_{\mc F})_{- 2\epsilon/P(\mc F)}^{-\max}}_{2/3}\}$ samples.
\end{lemma}

\begin{proof}[Proof of Lemma~\ref{lem:postselection}]
  For any $\mc F \subset \mc E$ we have
\begin{align}
  \norm{ P_{- \epsilon}^{- \max}}_{2/3} & \geq \norm{(P_{- \epsilon}^{- \max})_{\restriction \mc F}}_{2/3} \\
  &\geq \norm{(P_{\restriction \mc F})_{- \epsilon}^{- \max}}_{2/3} \\
  & = P(\mc F) \norm{ (P_{\restriction \mc F})_{- \epsilon}^{- \max} /P(\mc F)}_{2/3} \\
   & = P(\mc F) \norm{(P_{\mc F})_{- \epsilon/P(\mc F)}^{- \max}}_{2/3}
  .
\end{align}
Here, the first inequality becomes an equality in case $\mc F$ contains the support of $P_{-\epsilon}^{-\max}$.
The second inequality becomes an equality whenever the smallest probabilities with weight not exceeding $\epsilon$ as well as the largest probability lie inside of $\mc F$.
Finally, the last equality follows from the fact that when renormalizing $P_{\restriction \mc F}$ we also need to renormalize the subtracted total weight $\epsilon$ by the same factor.
The claim then straightforwardly follows from Theorem~\ref{thm:optimal tests}.
\end{proof}

A non-trivial bound for the sample complexity is therefore achieved only in case the post-selected subspace has at least weight $P(\mc F) > 2 \epsilon$.
This is due to the strength of Valiant and Valiant's result~\cite{valiant_automatic_2017} in that a part of the distribution with total weight $ 2 \epsilon$ does not influence the minimally required sample complexity of $\epsilon$-certification and this part might just be supported on $\mc F$.


\section{``Quantum supremacy'' distributions cannot be certified}
\label{sec:supremacy}

We will now apply the result of the previous section to the case of certifying ``quantum supremacy'' distributions.
As a result, we find that prominent schemes aimed at demonstrating ``quantum supremacy'', most importantly boson sampling, cannot be certified from polynomially many classical samples and a description of the target distribution alone.
To be more concrete, in the context of ``quantum supremacy'', there has recently been an enormous activity aiming to devise simple sampling schemes, which show a super-polynomial speedup over any classical algorithm even if the sampling is correct only up to a constant $\ell_1$-norm error \cite{aaronson_computational_2010,bremner_average-case_2016,boixo_characterizing_2016,miller_quantum_2017,gao_quantum_2017,Supremacy,bouland_complexity_2018,morimae_hardness_2017,hangleiter_anticoncentration_2018}.
The method used to prove all the aforementioned speedups is the proof technique pioneered by Terhal and DiVincenzo~\cite{terhal_adaptive_2004} for the case of exact sampling, which is based on an application of Stockmeyer's approximate counting algorithm \cite{Stockmeyer85ApproxiationSharpP}.
Stockmeyer's algorithm is used to prove that, conditioned on a conjecture on the average-case hardness of certain problems, the polynomial hierarchy would collapse if the respective distribution could be sampled efficiently classically.
To extend this proof technique to the case of approximate sampling up to an (additive) $\ell_1$-norm error requires an additional property on the sampled distribution, namely, \emph{anti-concentration} \cite{aaronson_computational_2010,bremner_average-case_2016}.

More precisely, the aforementioned tasks all fit the following schema:
Given the problem size $n$, start from a reference state vector $\ket{S_0}$ from a Hilbert space $\mc H_n$ and apply a unitary $U$ drawn with respect to some measure $\mu_n$ on the corresponding unitary group.
The resulting state is then measured in the computational basis, thereby resulting in outcome $S$ with probability $P_U(S) \coloneqq |\bra{S} U \ket{S_0}|^2$.
One then says that the distribution over $P_U$ induced by this procedure \emph{anti-concentrates} if
\begin{equation}
  \begin{split}
    \exists \, \alpha, & \gamma > 0:\ \forall n \in \mathbb{Z}_+ \text{ and } \forall S \in \mc E_n:  \\
    &\Pr_{U \sim \mu_n} \left( P_U(S) \geq \frac{\alpha}{|\mc E_n|} \right) \geq \gamma \, .
  \end{split}
  \label{eq:anticoncentration}
\end{equation}

In words this roughly means: For any $n$ the induced distribution over $P_U$ has the property that for any fixed outcome $S$ it does not become too unlikely that the probability of getting that outcome is much smaller than it would be for the uniform distribution.

It is intuitive that due to normalization, anti-concentration also implies that not too much of the probability weight can be concentrated in few outcomes, hence the name.
This is however slightly misleading, as anti-concentration does not in itself imply a that $P_U$ also needs to have a high min-entropy with high probability.
To see this, take any $\alpha,\gamma(\alpha)$-anti-concentrating scheme of drawing probability distributions $P_U$.
Construct $\tilde P_U$ from $P_U$ by dividing all probabilities in half and adding their joint weight to $P_U(0)$.
The resulting scheme to construct $\tilde P_U$ is now still $\alpha/2,\gamma(\alpha)$-anti-concentrating, but has min-entropy $\Hmin(\tilde P_U) \leq \log(2)$ with probability one.

However, most known proofs of anti-concentration \cite{bremner_average-case_2016,bremner_achieving_2017,aaronson_bosonsampling_2013,hangleiter_anticoncentration_2018} (with the notable exception of Morimae's hardness result~\cite{morimae_hardness_2017}\footnote{There, Morimae proves anti-concentration for the output distribution so-called one-clean qubit model (DQC1) in a direct manner. }) rely on the Paley-Zygmund inequality
\begin{align}
    \Pr ( Z > a \mathbb{E}[Z] ) \geq (1- a)^2 \frac{\mathbb{E}[Z]^2}{\mathbb{E}[Z^2]}\, ,
    \label{eq:paley-zygmund}
\end{align}
for a random variable $Z \geq 0$ with finite variance and $0 \le a \le 1$.
Anti-concentration is then proved by deriving a bound on the second moment of the distribution $\{P_U (S)\}_{U \sim \mu_n}$, and, as we will see in the next section, the same second moment bound that is used to derive anti-concentration implies a high min-entropy. We lay out the proof structure and the role that second moments play in Figure~\ref{fig:proof overview}.

\subsection{Second moments bound the min-entropy}

\label{sec:second moments}

We now turn to showing that with high probability over the choice of $U$, the second moment of the distribution $\{ P_U(S)\}_{U \sim \mu_n} $ (for fixed $S$) --- implying the anti-concentration property \eqref{eq:anticoncentration} --- yields a lower bound on the min-entropy of the output distribution $P_U$ of any fixed unitary $U$ with high probability.
Lemma~\ref{lem:bounds} then implies that distributions with exponentially (in $n$) small second moments cannot be certified from polynomially many samples with high probability.
Thus, the \emph{very property} that implies sampling hardness of a distribution $P_U$ up to an additive total-variation distance error also implies that the same distribution cannot be efficiently certified from classical samples only.
\begin{lemma}[Tail bound for the min-entropy]
\label{lem:anticon-minentropy}
	For any $n \in \mb Z^+$, let $P_U$ be a distribution on $\mc E_n $ induced via $P_U(S) = | \bra S U \ket{S_0} |^2$, $U \sim \mu_n$ by a corresponding measure $\mu_n$ on the unitary group.
	Then, with probability at least $1-\delta$ over the choice of $U \sim \mu_n $,
	\begin{align}
          \Hmin(P_U) \geq \frac12 \left( \log  \delta - \log \sum_{S \in \mc E_n} \mathbb{E}_{U\sim \mu_n}[P_U(S)^2] \right).
          \label{eq:anticon-minentropy}
	\end{align}
\end{lemma}
The following arguments used to derive the lemma, in fact, hold for more general families of probability distributions $P_U$ where $U$ need not be unitary without any scaling in $n$.

\begin{proof}[Proof of Lemma~\ref{lem:anticon-minentropy}]
We proceed as follows: First, we prove a lower bound on the typical R\'enyi-2-entropy of $P_U$ using the second moment of $\{P_U\}_{U \sim \mu_n}$, and then use equivalence of the $\alpha$-R\'enyi entropies for $\alpha >1$.

Analogously to Ref.~\citep[App.\ 11]{aaronson_bosonsampling_2013},
we use Markov's inequality to obtain that with probability at least $1-\delta$ over the choice of $U$, we have
\begin{align}
	H_2(P_U) & \coloneqq - \log \sum_{S \in \mc E_n} P_U(S)^2 \\
	& \geq - \log \left( \frac{1}{\delta} \cdot \mathbb{E}_{U \sim \mu_n } \left[\sum_{S \in \mc E_n} P_U(S)^2
	\right] \right) .
\end{align}
What is more, one can show that all $\alpha$-R{\'e}nyi entropies for $\alpha > 1$ are essentially equivalent and in particular
\cite{wilming_entanglement-ergodic_2018}
\begin{align}
 	H_\alpha(P) \geq \Hmin(P) \geq \frac{\alpha-1}{\alpha} H_\alpha(P) ,
 	\label{eq:entropy equivalence}
\end{align}
for any distribution $P$ on $\mc E_n$ from which the claim follows\footnote{We show Eq.~\eqref{eq:entropy equivalence} as well as an alternative proof for Lemma~\ref{lem:anticon-minentropy} in Section~\ref{app:min entropy}.}.
\end{proof}

We note that, indeed, the notion of anti-concentration as formalized in Eq.~\eqref{eq:anticoncentration} in itself does not necessarily imply that the output distribution of every (or most) fixed unitaries have high min-entropy.
This is because anti-concentration merely requires that the tails of the distribution have sufficient (constant) weight, while allowing for few large probabilities.
Nevertheless, in prominent cases, an anti-concentration result derives from bounds on the $2$-R\'enyi entropy.

\subsection{``Quantum supremacy'' distributions are flat}
\label{sec:supremacy_flat}

We now apply our results to the most prominent examples of ``quantum supremacy'' schemes --- boson sampling  \cite{aaronson_computational_2010}, IQP circuits \cite{bremner_average-case_2016}, and universal random
circuits \cite{boixo_characterizing_2016}.
We will conclude from Lemmas~\ref{lem:bounds}--\ref{lem:anticon-minentropy} that these schemes cannot be efficiently certified from polynomially many samples only.
In the following, we show that for all of these schemes with output distribution $P_U$ we have that $\Hmin(P_U) \in \Omega(n)$ and hence that the minimal sample complexity for certification scales exponentially in $n$.
More precisely, both for boson sampling  and the qubit-based schemes mentioned above all of which we precisely define below in Sections~\ref{sec:iqp}--\ref{sec:Boson Sampling}, we obtain the following lower bounds.

\begin{theorem}[Lower bounds on certifying boson sampling]
\label{thm:samplecomplexitybounds boson sampling}
	Let $0 < \epsilon < 1/2$, $n \in \mb Z_+$ sufficiently large and $m \in \Theta(n^\nu)$.
	Under the conditions on $\nu$ used in Ref.~\cite{aaronson_computational_2010} to prove the hardness of approximate boson sampling, and with high probability over the random choice of the unitary, there exists no $\epsilon$-certification test of boson sampling  with $n$ photons in $m$ modes from $s < s_{\min}$ many samples, where
	\begin{align}
		s_{\min} \in \Omega\left( 2^n / \epsilon^2 \right) \, .
	\end{align}
\end{theorem}

In Section~\ref{sec:Boson Sampling}, we discuss in detail the conditions under which Aaronson and Arkhipov's hardness argument \cite{aaronson_computational_2010} holds and provide a full version of the theorem as Theorem~\ref{theorem:samplecomplexitybounds boson sampling technical}.
The key ingredient for this to be the case is the closeness of the measure obtained by taking $n\times n $-submatrices of Haar random unitaries $U \in U(m)$ and the Gaussian measure on $n \times n$-matrices.
This is provably the case for
$\nu > 5$, but is conjectured to hold even for $\nu > 2$ \cite{aaronson_computational_2010}.
Our bound on $s_{\min}$ (see Theorem \ref{theorem:samplecomplexitybounds boson sampling technical}) holds with exponentially high probability (in $n$) for $\nu > 3$.
In the case $\nu > 2$ our result holds only with polynomially high probability and fails to cover a small set of the instances.
The argument proving Theorem~\ref{thm:samplecomplexitybounds boson sampling} easily extends to certain variants of
\emph{quantum Fourier sampling} \cite{fefferman_power_2015},
 the output probabilities of which are also given by permanents of nearly Gaussian matrices.

\begin{theorem}[Lower bounds on certifying random qubit schemes]
\label{thm:samplecomplexitybounds qubits}
For $0 < \epsilon < 1/2$ and sufficiently large $n$, with probability at least $1 - \delta$, there exists no $\epsilon$-certification test from $s < s_{\min}$ many samples for
\begin{enumerate}[a.]
	\item IQP circuit sampling on $n$ qubits, where
	\begin{align}
		s_{\mr{min}}
		&\in
		\Omega\left(  2^{n/4} \delta^{1/4} /\epsilon^2 \right) \, .
	\end{align}
	\item $\tilde{\varepsilon}$-approximate spherical 2-design sampling on $n$ qubits, and in particular, depth-$(O(n^2) + O(n \log 1/\tilde{\varepsilon}))$ local random universal circuits, where
	\begin{align}
		s_{\mr{min}}
		&\in
		\Omega\left(  \frac{2^{n/4}\delta^{1/4}}{\epsilon^2 (1+ \tilde{\varepsilon})^{1/4} } \right) .
	\end{align}
\end{enumerate}
\end{theorem}

The result of Theorem~\ref{thm:samplecomplexitybounds qubits} applies to any circuit family $\mc U$ such that $\{ U \ket{S_0} \}_{U \sim \mc U} $ forms a relative $\tilde{\varepsilon}$-approximate spherical 2-design, for which the second moments are upper bounded as in Eq.~\eqref{eq:2-design bound}.
This applies, in particular, to the random universal
circuits of Refs.~\cite{brandao_local_2016,harrow_random_2009,boixo_characterizing_2016,bouland_quantum_2018} as well as other families of random circuits that have been proposed for the demonstration of ``quantum supremacy'' such as Clifford circuits with magic-state inputs \cite{hangleiter_anticoncentration_2018,yoganathan_quantum_2018},
diagonal unitaries \cite{hangleiter_anticoncentration_2018,nakata_generating_2014} and
conjugated Clifford circuits \cite{bouland_complexity_2018}.

\begin{proof}[Proofs of Theorems~\ref{thm:samplecomplexitybounds boson sampling} and \ref{thm:samplecomplexitybounds qubits}]
We use Theorem~\ref{thm:optimal tests} and Lemmas~\ref{lem:bounds}--\ref{lem:anticon-minentropy}
as well as the lower bounds
\eqref{eq:nu3flatness},
\eqref{eq:IQP_Hmin_bound}, and
\eqref{eq:random_circuit_Hmin_bound} on the min-entropy of the respective output distributions as given in the following sections.
What is more, we use that for $0 <\epsilon < 1/2$ and sufficiently large $n$ the term $(1-2\epsilon-2^{-\Hmin(P_U)})^{3/2}$
can be lower-bounded by a constant and, hence, be dropped inside the $\Omega$.
\end{proof}

\section{Details on random sampling schemes and proofs of Theorems~\ref{thm:samplecomplexitybounds boson sampling} and \ref{thm:samplecomplexitybounds qubits} }

We now turn to describing details on distributions arising from boson sampling, IQP circuits and universal random circuits, and present
the proofs of the aforementioned theorems.

\subsection{IQP circuits}
\label{sec:iqp}

An IQP circuit~\cite{shepherd_temporally_2009} is a  quantum circuit of commuting gates that is drawn uniformly at random from the family $\mathcal{U}_{n,\mathrm{IQP}}$ on $n$ qubits.
The sample space is therefore given by $\mc E_n = \{ 0,1\}^n$.
This family as formulated by Bremner, Montanaro, and Shepherd~\cite{bremner_average-case_2016} is defined by a set of angles $A$, e.g., $A = \{0,\pi/8, \ldots, 7\pi/8\}$.
An instance $U_{W} \in \mathcal{U}_{n,\mathrm{IQP}}$ with $W \coloneqq (w_{i,j})_{i,j=1,\ldots, n}$ and $w_{i,j} \in A$ drawn uniformly at random, is then given by the following prescription
\begin{equation}
U_{W} = \exp \left[\mr i \left (\sum_{i < j} w_{i,j} X_i X_j + \sum_i w_{i,i} X_i \right) \right] \, ,
\label{eq:iqp circuits}
\end{equation}
where $X_i$ is the Pauli-$X$ matrix acting on site $i$.
In other words, on every edge $(i,j)$ of the complete graph on $n$ qubits a gate $\exp(\mr i w_{i,j} X_i X_j) $ with edge weight $w_{i,j}$ and on every vertex $i$ a gate $\exp(\mr i w_{i,i} X_i) $ with vertex weight $w_{i,i}$ is performed.

For the output distribution of IQP circuits, Bremner \emph{et al.}~\cite[Appendix~F]{bremner_average-case_2016} prove the second-moment bound
\begin{equation}
	\Eb_{W} [|\bra{S} U_W \ket{0}|^4 ] \leq 3 \cdot 2^{-2n} \, .
\end{equation}
By Lemma~\ref{lem:anticon-minentropy}, this implies the following min-entropy bound
\begin{equation}\label{eq:IQP_Hmin_bound}
	\Hmin(\piqp) \geq \frac12 \left(n + \log \frac{\delta}{3}\right) ,
\end{equation}
which holds with probability at least $ 1-\delta$ over the choice of $U_W$.

\subsection{Universal random circuits and spherical 2-designs}
\label{sec:universal circuits}

A universal random circuit on $n$ qubits is defined by a universal gate set $\mc G$ comprising one- and two-qubit gates which give rise to the depth-$N$ family $\mc U_{\mc G,N}$.
A circuit $U \in \mc U_{\mc G,N}$ is then constructed according to the standard prescription of choosing one- or two-qubit gates $G \in \mc G$ and the qubits they are applied to at random \cite{brandao_local_2016}, or according to some more specific prescription such as the one of~\citet{boixo_characterizing_2016}.

For the case of the random universal circuits of Ref.~\cite{boixo_characterizing_2016} there is evidence that the output distribution of fixed instances is essentially given by an exponential (Porter-Thomas) distribution $P_{\mr {PT}}$ whose second moment is given by \cite[Eq.~(8)]{hangleiter_anticoncentration_2018}
\begin{align}
  \label{eq:2ndmomentporter-thomas}
 	\mb E _{p \sim P_{\mr {PT}}} [p^2]= \frac 2 {|\mc E_n|( |\mc E_n| + 1 )} \, .
\end{align}
This is provably true for the local random universal circuits investigated by Brand{\~a}o,  Harrow, and Horodecki~\cite{brandao_local_2016}
by the fact that the resulting circuit family forms a relative $\tilde{\varepsilon}$-approximate unitary 2-design $\mu$ in depth $O(n^2) + O(n \log 1/ \tilde{\varepsilon})$ \cite{brandao_local_2016} so that
\begin{align}
\label{eq:2-design bound}
\mathbb{E}_{U \sim \mu} [|\bra{S} U \ket{S_0}|^4] \leq  \frac{2(1 + \tilde{\varepsilon})}{|\mc E_n|(|\mc E_n|+1)}.
\end{align}
Likewise, for any circuit family $\mc U_n$ on $n$ qubits such that $\{ U \ket{0} \}_{U \sim \mc U_n} $ forms a relative $\tilde{\varepsilon}$-approximate spherical 2-design, the second moments are bounded as in Eq.~\eqref{eq:2-design bound}.
For all such circuit families, using Lemma~\ref{lem:anticon-minentropy}, we thus obtain the min-entropy bound
\begin{align}\label{eq:random_circuit_Hmin_bound}
	\Hmin(P_U) \geq \frac12 \left( n + \log \frac{\delta}{2( 1 + \tilde{\varepsilon})}\right ) \, ,
\end{align}
 which holds with probability at least $ 1-\delta$ over the choice of $U$.

\subsection{Boson sampling}
\label{sec:Boson Sampling}

In the boson sampling  problem $n \geq 1$ photons are injected into the first $n$ of $m \in \poly(n)$ modes which are transformed in a linear-optical network via a mode transformation given by a Haar-random unitary $U \in U(m)$ and then measured in the Fock basis.
The sample space of boson sampling  is given by
\begin{equation} \label{eq:bosonsamplingsamplespace}
  \mc E_n \coloneqq \Phi_{m,n} \coloneqq \Big\{ (s_1,\dots,s_m) : \sum_{j=1}^m s_j = n \Big\} ,
\end{equation}
i.e., the set of all sequences of non-negative integers of length $m$ which sum to $n$.
Its output distribution $\pbos$ is
\begin{equation} \label{eq:bosonsamplingdistribution}
  \pbos(S) \coloneqq |\bra S \varphi(U) \ket{1_n} |^2 .
\end{equation}
Here, the state vector $\ket S$ is the Fock space vector corresponding to a measurement outcome $S \in \Phi_{m,n}$, $\ket{1_n}$ is the initial state vector with $1_n \coloneqq (1,\dots, 1,0,\dots, 0)$, and $\varphi(U)$ the Fock space (metaplectic) representation of the implemented mode transformation $U$.

The distribution $\pbos$ can be expressed \cite{scheel_permanents_2004,aaronson_computational_2010} as
\begin{equation} \label{eq:bosonsamplingdistribution permanent}
  \pbos(S) = \frac{|\Perm(U_S)|^2}{\prod_{j=1}^m (s_j!)},
\end{equation}
in terms of the permanent of the matrix $U_S \in \mb C^{n\times n}$ constructed from $U$ by discarding all but the first $n$ columns of $U$ and then, for all $j \in [m]$, taking $s_j$ copies of the $j^{\text{th}}$ row of that matrix (deviating from Aaronson and Arkhipov's notation~\cite{aaronson_computational_2010}).
Here, the permanent for a matrix $X = (x_{j,k}) \in \mb C^{n \times n }$ is defined similarly to the determinant but without the negative signs as
\begin{equation} \label{eq:definitionpermanent}
  \Perm(X) \coloneqq \sum_{\tau \in \Sym([n])} \prod_{j=1}^n x_{j,\tau(j)},
\end{equation}
where $\Sym([n])$ is the symmetric group acting on $[n]$.
It is a known fact that calculating the permanent of a matrix to high precision is a problem that is \sharpP-hard \cite{valiant_complexity_1979}, while its close cousin, the determinant, is computable in polynomial time.
In fact, computing the permanent exactly (or with exponential precision) is also \sharpP-hard \emph{on average} for randomly chosen Gaussian matrices \cite{lipton_permanent_1991,aaronson_computational_2010}.
In Ref.~\cite{aaronson_computational_2010} this connection is exploited to show that, up to plausible complexity-theoretic conjectures, approximately sampling from the boson sampling  distribution is classically intractable with high probability over the choice of $U$ if $m$ is scaled appropriately with~$n$.

The main part of the hardness proof of Ref.~\cite{aaronson_computational_2010} is to prove the classical hardness of sampling from the \emph{post-selected boson sampling  distribution} $\pbos^*$.
The post-selected distribution $\pbos^*$ is obtained from $\pbos$ by discarding all output sequences $S$ with more than one boson per mode, i.e., all $S$ which are not in the set of \emph{collision-free} sequences
\begin{equation}
  \Phi^*_{m,n} \coloneqq \Big\{ S \in \Phi_{m,n} : \forall s \in S : s \in \{0,1\} \Big\} .
\end{equation}
The hardness of sampling from the full boson sampling  distribution follows from the fact that for the relevant scalings of $m$ with $n$ the post-selection can be done efficiently in the sense that on average at least a constant fraction of the outcome sequences is collision-free (Theorem~13.4 in Ref.~\cite{aaronson_computational_2010}).

More precisely, the actual result proved in Ref.~\cite[Theorem~1.3]{aaronson_computational_2010} states that unless certain complexity-theoretic conjectures fail, there exists no classical algorithm that can sample from a distribution $Q$ satisfying $ \norm{Q - \pbos}_1 \leq \epsilon$ in time $\poly(n, 1/\epsilon)$.
This result requires that $m \in \Omega(n^5 \log(n)^2)$, but it is conjectured that $m$ growing slightly faster than $\Omega(n^2)$ is sufficient for hardness.
In fact, at the same time, a faster than quadratic scaling is necessary for the proof strategy to work.

The key technical ingredient in the proof strategy underlying these requirements is the following result:
if $m$ grows sufficiently fast with $n$, the measure induced on $U \sim \mu_H$ by the map $g_S = (U \mapsto U_S)$ for collision-free $S \in \Phi_{m,n}^*$, i.e., the measure induced by taking $n \times n$-submatrices of unitaries $U \in U(m)$ chosen with respect to the Haar measure $\mu_H$ is close to the complex Gaussian measure $\mu_G(\sigma)$ with mean zero and standard deviation $\sigma = 1/\sqrt m$ on $n \times n$-matrices.
Given this result, Stockmeyer's algorithm could be applied to the samples obtained from $\pbos^*$ in order to infer the probabilities $\pbos^*(S)$ and thus solve a \sharpP-hard problem, as these probabilities can be expressed as the permanent of a Gaussian matrix.
Since the closeness of those measures is the essential ingredient, also
suitably large scaling of $m$ with $n$ is crucial for the hardness argument.

The formal statement of closeness of measures proved in Ref.~\cite{aaronson_computational_2010} implies the following:
\begin{lemma}[implied by Ref.~{\cite[Theorem~5.2]{aaronson_computational_2010}}]\label{lemma:multiplicativeerrorbound}
 There exists a constant $C > 0$ such that for every $\nu>5$ and every measurable $f:\mb C^{n\times n} \to \left[0,1\right]$ and every $m \in \Omega(m^\nu)$ it holds that for all $S \in \Phi^*_{m,n}$
  \begin{equation} \label{eq:multiplicativeerrorbound1}
    \quad \Eb_{U \sim \mu_H} f(U_S)  \leq (1+C)\, \Eb_{X \sim \mu_{G(1/\sqrt{m})}} f(X) .
  \end{equation}
\end{lemma}
At the same time, it is known from Ref.~\cite{Jiang2006} (see also Ref.~\cite[Section 5.1 and 6.2]{aaronson_computational_2010})
that if $m \leq c\,n^{\nu}$ with $\nu \leq 2$ and $c \in O(1)$ the two measures $\mu_H \circ g_S^{-1}$ and $\mu_{G(1/\sqrt{m})}$ are no longer close for large $n$.
One may hope \cite{aaronson_computational_2010} that there exists a constant $c>0$ such that Theorem~5.2 in Ref.~\cite{aaronson_computational_2010} and hence their hardness result as well as our Lemma~\ref{lemma:multiplicativeerrorbound} hold for any $m \geq c\,n^{\nu}$ with $\nu > 2$.
What we show is that \emph{even under this optimistic assumption} efficient certification from classical samples is impossible, if the post-selection probability is large enough.
This rules out many further cases for which one can hope to prove a hardness result by the same method.

\begin{theorem}[Lower bounds on certifying boson sampling  (full version)]
\label{theorem:samplecomplexitybounds boson sampling technical}
  Let $\nu > 2$,  $0 < \epsilon < 1/2$, $n \in \mb Z_+$ sufficiently large and $m \in \Theta(n^\nu)$.
  Assume there exists a constant $C > 0 $ such that
  the assertion \eqref{eq:multiplicativeerrorbound1} of Lemma~\ref{lemma:multiplicativeerrorbound} holds.
  Then:
\begin{enumerate}[a.]
\item \label{item:a}
	With probability at least $1-\delta - 2n^2/(m\zeta)$ over the choice of Haar-random unitaries $U \sim \mu_H$ there exists no $\epsilon$-certification test for boson sampling  with $n$ photons in $m$ modes,
	from $s < s_{\mr {min}}$ many samples,
		where
    \begin{align}
    \label{eq:smin_bs_ps_full}
   s_{\min} \in  \Omega\left( n^{c n (\nu - 1)/4} \delta^{1/4}  \left( 1- \zeta  - 2\epsilon \right)^{3/2}/\epsilon^2 \right) , 
		\end{align}
		and $c>0$ is the implicit constant in \eqref{eq:second moment bound min entropy bs}.
\item \label{item:b}
	For $\nu > 3$, with probability at least
	$1-\exp(-\Omega(n^{\nu-2-1/n}))$ over the Haar-random choice of $U \sim \mu_H$,
	there exists no $\epsilon$-certification test for boson sampling  with $n$ photons in $m$ modes,
		where
		\begin{align}
		  s_{\mr {min}} \in \Omega\left( 2^n/\epsilon^2 \right) .
		\end{align}
\end{enumerate}
\end{theorem}
We remark that our results for the boson sampling distribution leave open the possiblity of sample-efficient $\epsilon$-certification for those instances of boson sampling with $2 < \nu \le 3$ in the regime in which the probability weight of the collision-free subspace is very small. 
For instance, this is the case whenever $1/\poly(n) \le \pbos(\Phi_{m,n}^*) \le 2 \epsilon$.
This is because the bound \eqref{eq:smin_bs_ps_full} becomes trivial for $1 - \zeta \leq 2 \epsilon$.

However, our result fully covers the regime in which boson sampling is provably hard as shown in Ref.~\cite{aaronson_computational_2010}.

\begin{proof}[Proof of Theorem~\ref{theorem:samplecomplexitybounds boson sampling technical}\ref{item:a}]
The proof proceeds along the same lines as the proofs of Theorem~\ref{thm:samplecomplexitybounds qubits} and is based on direct applications of Lemma~\ref{lem:postselection} to the collision-free subspace and the min-entropy bound~\eqref{eq:anticon-minentropy} 
from Lemma~\ref{lem:anticon-minentropy}
to the post-selected boson sampling distribution $\pbos^*$ with post-selection onto the collision-free subspace $\Phi_{m,n}^*\subset \Phi_{m,n}$.

To apply Lemmas~\ref{lem:postselection} and \ref{lem:anticon-minentropy} simultaneously we need to account both for the probability weight of the collision-free subspace and large probabilities, however.
To account for the probability weight of the collision-free subspace we use a simple application of Markov's inequality to \cite[Theorem~13.4]{aaronson_computational_2010} (restated as Lemma~\ref{lem:sizecollisionfree} in Section~\ref{app:collision free}), 
\begin{equation}
\Pr_{U \sim \mu_H}\left[ \pbos[\Phi_{m,n}\setminus\Phi_{m,n}^*] > \zeta \right]
< \frac{2n^2}{\zeta m} \, .
\end{equation}
This shows that the total probability weight of the collision-free subspace is at least $1- \zeta$ with probability at least $1 - 2n^2/(m\zeta)$.
We then apply a union bound argument to obtain
\begin{align}\label{eq:union_bound_ps}
\begin{split}
  \Pr_{U\sim \mu_H}
    \bigg[& \left\{\pbos \text{ does not satisfy \eqref{eq:anticon-minentropy}} \right\} \\\
    &
     \cup
     \left \{\pbos(\Phi_{m,n}\setminus \Phi_{m,n}^*) > \zeta \right\} \bigg]
  \leq
  \delta + \frac{2n^2}{\zeta m} \, . 
  \end{split}
\end{align}
In the next step, we use that the distribution of post-selected boson sampling is given by $\pbos^* = (\pbos)_{\restriction \Phi_{m,n}^*} / \pbos(\Phi_{m,n}^*)$. 
Consequently, with probability at least $1-\delta-{2n^2}/({\zeta m})$ the boson sampling distribution $\pbos$ restricted to the collision-free subspace has both of the desired properties -- a large min-entropy and a probability weight of at least $1 - \zeta$ of the collision-free subspace. 

Let us now compute the min-entropy for the collision-free subspace. 
For all samples $S \in \Phi_{m,n}^*$, Ref.~\cite[Lemma~8.8]{aaronson_computational_2010} implies that there exists $C>0$ such that for $m \in \Theta(n^\nu)$ with any $\nu >2$ for which the assertion of Lemma~\ref{lemma:multiplicativeerrorbound} holds, the following second moment bound also holds\footnote{The version of Ref.~\cite[Lemma~8.8]{aaronson_computational_2010} can be obtained from Eq.~\eqref{eq:second moments bs} from Lemma~\ref{lemma:multiplicativeerrorbound}, normalizing the Gaussian measure $\mu_G$, and noting that $\Eb_{X \sim \mu_G(1)} [|\Perm(X)|^2] = n!$.}:
\begin{equation}
  \Eb_{U_S \sim \mu_H}[|\Perm(U_S)|^4] \leq ( 1 + C )(n!)^2\,(n+1)\,m^{-2n} \, .
  \label{eq:second moments bs}
\end{equation}

To obtain a lower bound on the min-entropy of the distribution $\pbos^*$ on the collision-free subspace we use that 
\begin{align}
  \Hmin(\pbos^*) = \log \pbos(\Phi_{m,n}^*) +  \Hmin((\pbos)_{\restriction \Phi_{m,n}^*}) . 
\end{align}
Applying Lemma~\ref{lem:anticon-minentropy} together with the second moment bound \eqref{eq:second moments bs}, the union bound \eqref{eq:union_bound_ps}, the bound 
  \begin{align}
    |\Phi_{m,n}^*| = \binom m n = \frac{ m (m-1) \cdots (m-n + 1)}{n!} \leq \frac {m^n}{n!},
  \end{align}
on the size of the collision-free subspace and Stirling's formula yields
\begin{align}
\begin{split}
    2\Hmin& (\pbos^*) \ge   2 \log \left(1- \zeta \right) + \log \delta\\
     & - \log \left(\frac{m^n}{n!} ( 1+ C )(n!)^2 (n+1) m^{-2n} \right) \\
     & \in \Omega\left((\nu - 1)n \log n\right)   -  \log \frac 1 \delta \, - 2 \log \frac 1 { 1- \zeta} ,
  \label{eq:second moment bound min entropy bs}
\end{split}
\end{align}
which holds with probability $1-\delta - 2n^2/(m\zeta)$ over the choice of $U \sim \mu_H$.

We note that $2^{- \Hmin(\pbos^*)} \in o(1)$; hence this term can be neglected when applying Eq.~\eqref{eq:2/3 bounds} in Lemma~\ref{lem:bounds}. 
Applying Lemmas~\ref{lem:bounds} and~\ref{lem:postselection}, and the min-entropy bound~\eqref{eq:second moment bound min entropy bs} we obtain that the sample complexity for $\epsilon$-certifying boson-sampling scales as 
\begin{align}
  s_{\min} \in  \Omega\left( n^{c n (\nu - 1)/4} \delta^{1/4}  \left( 1- \zeta  - 2\epsilon \right)^{3/2}/\epsilon^2 \right) 
  \label{eq:bound sample complexity bs proof}
\end{align}
with probability at least $ 1 - \delta - 2 n^2 / (\zeta m )$, where $c$ is the implicit constant in \eqref{eq:second moment bound min entropy bs}. 
This completes the proof of Theorem~\ref{theorem:samplecomplexitybounds boson sampling technical}\ref{item:a}.

\end{proof}

Note that the bound \eqref{eq:second moments bs} is essential for the hardness argument of \citet{aaronson_computational_2010}. 
Therefore a central ingredient to the hardness argument of \citet{aaronson_computational_2010} also prohibits sample-efficient certification of boson sampling. 

It is important to stress that the boson sampling hardness proof~\cite{aaronson_computational_2010} covers only those instances $U_S$ of boson sampling for which one can efficiently post-select on the collision-free outcomes.
This is the case for those $U \sim \mu_H$ for which the probability weight of $\Phi_{m,n}^*$ is not smaller than polynomially small in $n$, i.e., $\pbos(\Phi_{m,n}^*) \in \Omega(1/\poly(n))$.
Our proof method for Theorem~\ref{theorem:samplecomplexitybounds boson sampling technical}\ref{item:a} thus permits sample-efficient certification for a small fraction of the instances, in particular, those instances of $U \sim \mu_H$ for which  $2 \epsilon \geq \pbos(\Phi_{m,n}^*) > 1/\poly(n)$. 

In part \ref{item:b} of the theorem we can close this gap by extending the bound \eqref{eq:second moment bound min entropy bs} on the min-entropy of the post-selected distribution $\pbos^*$ to the full output distribution $\pbos$, however, at the cost of restricting to $\nu > 3$.
This removes the need to use Lemma~\ref{lem:postselection} and hence the dependence on the probability weight of the collision-free subspace.
In the remaining case with $2 < \nu \le 3$ hardness results have not been obtained, but it is conceivable that a hardness argument can be made.

\begin{proof}[Proof of Theorem~\ref{theorem:samplecomplexitybounds boson sampling technical}\ref{item:b}]
Gogolin et al.~\cite{gogolin_boson-sampling_2013} have proven the following strong lower bound on the min-entropy of the boson sampling  distribution
(see Theorem~\ref{thm:min_entropy_bound_bs} in Section~\ref{app:proofs bs} for a restatement)
\begin{equation}
	\Pr_{U \sim \mu_H}\left[ \Hmin(\pbos)< 2 \, n \right]  \in \exp\left(-\Omega(n^{\nu-2-1/n}) \right) ,
	\label{eq:nu3flatness}
\end{equation}
which holds whenever the condition of the theorem are fulfilled and in addition $\nu > 3$.
In the proof, the probability measure induced on the matrices $U_S$ is related to a certain Gaussian measure $\mu_{G_S(\sigma)}$.
Then, the min-entropy bound is proven using a trivial upper bound to the permanent as well as measure concentration for $\mu_{G_S(\sigma)}$.
A simple application of Theorem~\ref{thm:optimal tests} and Lemma~\ref{lem:bounds} concludes the proof.
\end{proof}

\ifarxiv
\section{Conclusion}
We have shown that probability distributions with a high min-entropy cannot be certified from polynomially many samples, even when granting the certifier unlimited computational power and a full description of the target distribution.
Our result applies to the problem of certifying quantum sampling problems as proposed to demonstrate a quantum speedup in a non-interactive device-independent fashion.
We discuss the ironic situation that the very property that crucially contributes to the proof of approximate sampling hardness via Stockmeyer's algorithm and the Paley-Zygmund inequality --- the second moments of the sampled distribution --- forbids sample-efficient classical verification.
Our results highlight the importance of devising more elaborate certification schemes that allow for interaction between certifier and prover, invoke further complexity-theoretic assumptions or such on the sampling device, and/or grant the certifier some small amount of quantum capacities.

\section*{Acknowledgements} 
We are grateful to Adam Bouland who pointed us to the literature on property testing and thus provided the missing clue for finishing this project.
We would like to thank Ashley Montanaro, Bill Fefferman, Tomoyuki Morimae, Martin Schwarz, and Juan Bermejo-Vega for fruitful discussions and Aram Harrow and Anand Natarajan for sharing an early version of their related work. 
Finally, we would like to thank two anonymous referees for their thorough proof-checking and interesting questions which helped improve the manuscript.

D.\ H.\ and J.\ E.\ acknowledge support from the ERC (TAQ), the Templeton Foundation, the DFG (EI 519/14-1, EI 519/9-1, EI 519/7-1, CRC 183), and the European Union’s Horizon 2020 research and innovation programme under grant agreement No.\ 817482 (PASQUANS).
C.\ G.\ acknowledges support by the European Union’s Marie Sk\l{}odowska-Curie Individual Fellowships
(IF-EF) programme under GA: 700140 as well as financial support from ARO under contract W911NF-14-1-0098 (Quantum Characterization, Verification, and Validation), the Spanish Ministry MINECO (National Plan 15 Grant: FISICATEAMO No.\ FIS2016-79508-P, SEVERO OCHOA No.\ SEV-2015-0522), Fundació Cellex, Generalitat de Catalunya (Grants No.\ SGR 874 and No.\ 875, CERCA Programme, AGAUR Grant No.\ 2017 SGR 1341 and CERCA/Program), ERC (CoG QITBOX and AdG OSYRIS), EU FETPRO QUIC, EU STREP program EQuaM (FP7/2007–2017, Grant No.\ 323714), and the National Science Centre, Poland-Symfonia Grant No.\ 2016/20/W/ST4/00314.

\fi

  \cleardoublepage
  \widetext
  \ifjournal
  \section*{Appendices to the Supplementary Material}
  \else
  \appendix
  \section*{Appendix}
  \fi
    In the following, we (re)state some facts and earlier
    results in order to make this work self-contained. 

\section{Proofs for bounding the min-entropy }
\label{app:min entropy}

Here, we provide some details and proofs to statements made in Section~\ref{sec:second moments}.
First, we show the equivalence of the R\'enyi entropies \eqref{eq:entropy equivalence} proceeding analogously to Ref.~\cite{wilming_entanglement-ergodic_2018}:
we simply use that for $\alpha > 1$ and $p_0 = \norm{P}_\infty$ we have $p_0^\alpha \leq \sum_i p_i^\alpha$. Hence,
\begin{align}
\frac{\alpha}{\alpha-1} \log (p_0) & \leq \frac{1}{\alpha-1} \log \sum_{i=0}^{|\mc E_n| -1} p_i^\alpha \\
 \Leftrightarrow
- \frac{\alpha}{\alpha-1} \Hmin(P) & \leq - H_\alpha(P)\\
\Leftrightarrow
 \Hmin(P) & \geq \frac{\alpha-1}{\alpha}H_\alpha(P) \, .
\end{align}
We also provide an alternative proof of Lemma~\ref{lem:anticon-minentropy} based on the proof of Ref.~\cite[Theorem~13]{gogolin_boson-sampling_2013}.
\begin{proof}[Alternative proof of Lemma~\ref{lem:anticon-minentropy}]
	We begin the proof by noting that
	\begin{align}
		\Pr_{U \sim \mu_n} \left[ \Hmin(P_U) \leq \log \frac 1 \delta \right] =\Pr_{U \sim \mu_n}\left[ \exists S\in\mc E_n: P_U(S) \geq \delta \right] .
	\end{align}
  Using the union bound (also known as Boole's inequality) we obtain that for every $\delta > 0$
	\begin{align}
	\label{eq:boolesinequalitybound}
	&\Pr_{U \sim \mu_n}\left[ \exists S\in\mc E_n: P_U(S) \geq \delta \right]
	\leq \sum_{S\in\mc E_n} \Pr_{U \sim \mu_n}\left[ P_U(S) \geq \delta \right] .
	\end{align}
	Next, using Markov's inequality we can bound
	\begin{align}
		\Pr_{U \sim \mu_n}\left[ P_U(S) \geq \delta \right] \leq \frac 1 {\delta^2} \mb E_{U \sim \mu_n }\left[ P_U(S)^2 \right] \, ,
	\end{align}
	which concludes the proof.
\end{proof}

\section{Probability weight of the collision-free subspace}
\label{app:collision free}

We recapitulate a bound of \citet{aaronson_computational_2010} on the probability weight of the collision-free subspace.
\begin{lemma}[{\cite[Theorem~13.4]{aaronson_computational_2010}}]
\label{lem:sizecollisionfree}
	Let $\mu_H$ be the Haar measure on $U(m)$ and $m \geq n $.
	Then
	\begin{align}
		\mb E_{U \sim \mu_H} \left[
			\pbos(\Phi_{m,n} \setminus \Phi_{m,n}^* )
		\right]
		\leq
		\frac {2 n^2} m \, .
	\end{align}
\end{lemma}

\section{The min-entropy bound for boson sampling}
\label{app:proofs bs}

Here, we provide a slightly improved proof of the following min-entropy bound for boson sampling from \cite[Theorem~12]{gogolin_boson-sampling_2013}.
\begin{theorem}[Min-entropy bound for boson sampling  {\cite[Theorem~12]{gogolin_boson-sampling_2013}}]\label{thm:min_entropy_bound_bs}
Let $\nu>3$ and assume that the assertion \eqref{eq:multiplicativeerrorbound1} of Lemma~\ref{lemma:multiplicativeerrorbound} holds.
Then, the boson sampling  output distribution $\pbos$ satisfies for $n$ bosons in $m \in \Theta(n^\nu)$ modes
\begin{equation}
  \Pr_{U \sim \mu_H}\left[ \Hmin(\pbos)< 2 \, n \right]
  \in
	\exp\left(-\Omega( n^{\nu-2-1/n}) \right) .
\end{equation}

\end{theorem}
The proof crucially uses the closeness of the Gaussian measure to the post-selected Haar measure as expressed by Lemma~\ref{lemma:multiplicativeerrorbound}.
Lemma~\ref{lemma:multiplicativeerrorbound}, however, is not quite strong enough for proving Theorem~\ref{theorem:samplecomplexitybounds boson sampling technical}, as we must be able to control all of $\Phi_{m,n}$ and not only the collision-free subspace $\Phi^*_{m,n}$.
Fortunately, the above lemma extends naturally to all $S \in \Phi_{m,n}$ for the same scaling of $m$ with $n$ for which a version of Lemma~\ref{lemma:multiplicativeerrorbound} holds.

To state this extension we need some notation first:
For every sequence $S$, let $\tilde S$ be the sequence obtained from $S$ by removing all the zeros, i.e,
\begin{equation}\label{eq:tildeSdef}
  \tilde S = (\tilde{s}_1,\dots, \tilde{s}_{|\tilde S|}) \coloneqq (s \in S: s > 0) .
\end{equation}
Further, let $\mu_{G_S(\sigma)}$ be the probability measure on $\mb C^{n \times n}$ obtained by drawing the real and imaginary part of every entry of a $|\tilde S| \times n$ matrix independently from a Gaussian distribution with mean zero and standard deviation $\sigma$ and then for all $j \in [|\tilde S |]$ taking $\tilde{s}_j$ copies of the $j^{\text{th}}$ row of this matrix.
We can prove the following multiplicative error bound on the closeness of this measure and the Haar measure $\mu_H$ for all $S \in \Phi_{m,n}$:
\begin{lemma}[Multiplicative error bound]\label{lemma:multiplicativeerrorbound2}
  Let $f:\mb C^{n\times n} \to \left[0,1\right]$ be measurable, then for any $m,n$ such that
    \begin{equation} \label{eq:multiplicativeerrorbound}
    \forall S \in \Phi^*_{m,n}:\quad \Eb_{U \sim \mu_H} f(U_S)  \leq (1+C)\, \Eb_{X \sim \mu_{G(1/\sqrt{m})}} f(X) ,
  \end{equation}
  is true for some constant $C >0$, it holds that
  \begin{equation} \label{eq:multiplicativeerrorbound2}
    \forall S \in \Phi_{m,n} :\quad \Eb_{U \sim \mu_H}f(U_{S}) \leq (1 + C)\, \Eb_{X \sim \mu_{G_S(1/\sqrt{m})}} f(X).
  \end{equation}
\end{lemma}
\begin{proof}
  Let $S \in \Phi_{m,n}$, define $\tilde S$ as in Eq.~\eqref{eq:tildeSdef} and $m'\coloneqq |\tilde S|$.
  Define $v$ to be the sequence containing $\tilde s_j$ times the integer $j$ for every $j \in [m']$ in increasing order and $w$ the sequence containing the positions of each of the first of the repeated rows in $U_S$, i.e.,
  \begin{align}
    v &\coloneqq (\underbrace{1,\ldots,1}_{\tilde s_1},
    \underbrace{2,\ldots,2}_{\tilde s_2},\ldots,\underbrace{m',\ldots, m'}_{\tilde s_{m'}})
    \in (\mb Z^+)^n,
    \\
    w &\coloneqq (1,1+\tilde s_1, 1+\tilde s_1+\tilde s_2, \ldots, 1+\sum_{j=1}^{m'-1} \tilde s_j)
    \in (\mb Z^+)^{m'}.
  \end{align}
  The sequence $v$ defines a linear embedding $\eta : \mb C^{m' \times n} \to  \mb C^{n \times n}$ component wise by
  \begin{align}
    \eta(Y)_{i,j} \coloneqq Y_{v_i, j} \quad \forall i, j \in [n],
  \end{align}
  i.e., $\eta(Y)$ has $s_j$ copies of the $j$-th row of $Y$.
  The sequence $w$, in turn, defines a linear projection $\pi : \mb C^{n \times n} \to \mb C^{m' \times n}$ by
  \begin{align}
    \pi(X)_{i,j} \coloneqq X_{w_i, j} \quad \forall i \in [m'],\ j \in [n],
  \end{align}
  in particular, $\pi(U_S)$ contains only the first out of each series of the repeated rows in $U_S$.
  Note that $\eta \circ \pi : \mb C^{n \times n} \to \mb C^{n \times n}$
  is a projection onto the subspace of matrices that have the same repetition structure as $U_S$.
  Let
  \begin{align}\label{eq:deffS}
    f_S \coloneqq f \circ \eta \circ \pi ,
  \end{align}
  then $f_S(U_S) = f(U_S)$ only depends on the first of the repeated rows in $U_S$ and is independent of all the other rows.
  Since the Haar measure is permutation-invariant we have
  \begin{align}
    \Eb_{U \sim \mu_H} f_S(U_S) &= \Eb_{U \sim \mu_H} f_S(U_{1_n}) .
  \end{align}
  Hence, using Lemma~\ref{lemma:multiplicativeerrorbound} in the second step, we obtain
  \begin{align}
    \Eb_{U \sim \mu_H} f(U_S)
    &= \Eb_{U \sim \mu_H} f_S(U_{1_n})
    \\
    &\leq (1 + C) \Eb_{X \sim \mu_{G(1/\sqrt m)}}f_S(X)
    \\
    &=(1 + C) \Eb_{X \sim \mu_{G_S(1/\sqrt m)}}f(X) ,
  \end{align}
  which finishes the proof.
\end{proof}

In addition to the multiplicative error bound we need the following concentration result for the Gaussian measure $\mu_{G_S(\sigma)}$, which implies that even the largest entry of a matrix drawn from $\mu_{G_S(\sigma)}$ is unlikely to be much larger than $\sigma$.
\begin{lemma}[Concentration of the Gaussian measure $\mu_{G_S(\sigma)}$] \label{lemma:concentrationofthegaussianmeasure}
For all $n,m \in \mb Z^+$, all $S \in \Phi_{m,n}$ and all $\xi > 0$ it holds that
\begin{equation}
  \Pr_{X\sim \mu_{G_S(\sigma)}}\left[ \max_{j,k\in\left[n\right]} |x_{j,k}| \geq \xi \right]
  \leq 1 - \left(1 - \Erfc\left(\frac{\xi}{\sqrt{2}\,\sigma }\right)\right)^{n^2} ,
\end{equation}
where
\begin{equation}
  \Erfc\left(\frac{\xi}{\sqrt{2}\,\sigma}\right) \coloneqq 2 \int_{\xi}^\infty \frac{\e^{-\frac{x^2}{2\,\sigma^2}}}{\sqrt{2\,\pi\,\sigma^2}} \dd x
\end{equation}
is the complementary error function.
\end{lemma}
\begin{proof}
  For Gaussian random variables we have
  \begin{equation}
    \forall \xi>0,\ j,k\in\left[n\right]:\quad \Pr_{X\sim \mu_{G(\sigma)}} \left[|x_{j,k}| \geq \xi \right] = \Erfc\left(\frac{\xi}{\sqrt{2}\,\sigma}\right) .
  \end{equation}
This implies that
\begin{equation}
\forall \xi>0:\quad \Pr_{X\sim \mu_{G(\sigma)}}\left[ \forall j,k\in\left[n\right]: |x_{j,k}| \leq \xi \right] = \left(1- \Erfc\left(\frac{\xi}{\sqrt{2}\,\sigma}\right)\right)^{n^2} .
\end{equation}
At the same time, for all $S \in \Phi_{m,n}$ and $ \xi>0$ it holds that
\begin{equation}
  \Pr_{X\sim \mu_{G_S(\sigma)}} \left[ \forall j,k\in\left[n\right]: |x_{j,k}|\leq  \xi \right] \geq
  \Pr_{X\sim \mu_{G(\sigma)}} \left[ \forall j,k\in\left[n\right]: |x_{j,k}| \leq \xi \right] ,
\end{equation}
because the repetition of entries in $X \sim \mu_{G_S(\sigma)}$ only increases the chance of not having an exceptionally large entry.
\end{proof}

As a last ingredient we need to bound the size
\begin{equation}
  |\Phi_{m,n}| = \binom{m+n-1}{n}
\end{equation}
of the sample space $\Phi_{m,n}$ of boson sampling  (recall Eq.~\eqref{eq:bosonsamplingsamplespace}).
It grows faster than than exponentially with $n$, but if for some $\nu\geq1$ and $c\geq0$ it holds that $m \leq c\,n^{\nu}$, then
\begin{align} \label{eq:samplespaceishuge}
  |\Phi_{m,n}| & \leq \frac{(m+n-1)^n}{n!} \leq \left(\frac{(m+n-1)\,\e}{n}\right)^n \\
  &
  \leq \e^n\,(c\,n^{\nu-1} + 1 - 1/n )^n \leq  (2\,(c+1)\,\e)^n\,n^{(\nu-1)\,n} . \label{eq:Phibound}
\end{align}

We now have all the ingredients rederive the desired min-entropy bound in Theorem~\ref{thm:min_entropy_bound_bs}.

\begin{proof}[Proof of Theorem~\ref{thm:min_entropy_bound_bs}]
  Using the union bound (also known as Boole's inequality) in the first step we obtain that for every $\epsilon > 0$
  \begin{align}
    &\Pr_{U \sim \mu_H}\left[ \exists S\in\Phi_{m,n}: \pbos(S) \geq \epsilon \right] \label{eq:boolesinequalitybound3}\\
    &\leq \sum_{S\in\Phi_{m,n}} \Pr_{U \sim \mu_H}\left[ \pbos(S) \geq \epsilon \right] \label{eq:boolesinequalitybound2}\\
    &\leq |\Phi_{m,n}| \max_{S\in\Phi_{m,n}} \Pr_{U \sim \mu_H}\left[ \pbos(S) \geq \epsilon \right] \label{eq:maxboundonsum} \\
    &= |\Phi_{m,n}| \max_{S\in\Phi_{m,n}} \Pr_{U \sim \mu_H}\left[ \frac{|\Perm(U_{S})|^2}{\prod_{j=1}^m (s_j!)} \geq \epsilon \right] .\label{eq:maximumprobabilitylargerepsilon}
  \end{align}
  We now apply Lemma \ref{lemma:multiplicativeerrorbound2} to the indicator function
  \begin{equation}
    f(U_S) =
    \begin{cases}
      1 & \text{if } \frac{|\Perm(U_{S})|^2}{\prod_{j=1}^m (s_j!)} \geq \epsilon \\
      0 & \text{otherwise}
    \end{cases},
  \end{equation}
  and the $S$ for which the maximum in Eq.~\eqref{eq:maximumprobabilitylargerepsilon} is attained, to obtain
  \begin{equation}
    \begin{split}
      &\Pr_{U \sim \mu_H}\left[ \exists S\in\Phi_{m,n}: \pbos(S) \geq \epsilon \right]
      \leq (1+C)\, |\Phi_{m,n}| \max_{S\in\Phi_{m,n}} \Pr_{X \sim \mu_{G_S(1/\sqrt{m})}} \left[ \frac{|\Perm(X)|^2}{\prod_{j=1}^m (s_j!)} \geq \epsilon \right] . \label{eq:lastinequalitybeforecudeboundonpermx}
    \end{split}
  \end{equation}
  The definition of the permanent (recall Eq.~\eqref{eq:definitionpermanent}) implies that
  \begin{equation} \label{eq:stupidpermamentbound}
    \frac{|\Perm(X)|^2}{\prod_{j=1}^m (s_j!)} \leq |\Perm(X)|^2 \leq (n!)^2\,\left(\max_{j,k \in \left[n\right]} |x_{j,k}|\right)^{2n} .
  \end{equation}
  Hence, for every $S\in\Phi_{m,n}$ and every $\epsilon > 0$
  \begin{equation}
    \Pr_{X \sim \mu_{G_S(1/\sqrt{m})}} \left[ \frac{|\Perm(X)|^2}{\prod_{j=1}^m (s_j!)} \geq \epsilon \right]
    \leq \Pr_{X \sim \mu_{G_S(1/\sqrt{m})}} \left[ \max_{j,k \in \left[n\right]} |x_{j,k}| \geq \left(\frac{\sqrt{\epsilon}}{n!}\right)^{1/n}   \right] .
  \end{equation}
  Plugging this into Eq.~\eqref{eq:lastinequalitybeforecudeboundonpermx}, using Lemma~\ref{lemma:concentrationofthegaussianmeasure} with $\xi = \left(\sqrt{\epsilon}/n!\right)^{1/n}$ and the bound on $|\Phi_{m,n}|$ from Eq.~\eqref{eq:Phibound} we arrive at
  \begin{equation}
    \begin{split}
      &\Pr_{U \sim \mu_H}\left[ \exists S\in\Phi_{m,n}: \pbos(S) \geq \epsilon \right] \\
      &\leq (1 + C)\,(2\,(c+1)\,\e)^n\,n^{(\nu-1)n} \left( 1 - \left(1 - \Erfc\sqrt{\frac{c}{2} \frac{\epsilon^{1/n}\,n^{\nu} }{(n!)^{2/n} }}\right)^{n^2} \right). \label{eq:boundonlargeprobabilitystillwitherfc}
    \end{split}
  \end{equation}
  Bounding the complementary error function by \cite{ermolova_simplified_2004}
    \begin{equation}
      \Erfc\left(x\right) \leq \e^{-x^2} ,
  \end{equation}
  we obtain
  \begin{align}
    1 - \left(1 - \Erfc(x) \right)^{n^2}
    &\leq 1 - \left(1 - \e^{-x^2} \right)^{n^2}
    = 1 - \sum_{k=0}^{n^2} \binom{n^2}{k}\,(-\e^{-x^2})^k \\
    &= \sum_{k=1}^{n^2} \binom{n^2}{k}\,\e^{-x^2 k}\,(-1)^{k-1}
    \leq \sum_{k=1}^{n^2} (n^2 \e/k)^k\,\e^{-x^2 k}  \\
    &\leq\sum_{k=1}^{n^2} (n^2\,\e^{-x^2 + 1} )^k . \label{eq:geometricseries}
  \end{align}
  If $x$ is large enough such that
  \begin{equation} \label{eq:convergencecondition}
    n^2\,\e^{-x^2+1} \leq \frac{1}{2} < 1,
  \end{equation}
  the geometric series in Eq.~\eqref{eq:geometricseries} converges and we get the simple bound
  \begin{equation} \label{eq:boundbygeometricseries}
    1 - \left(1 - \Erfc(x) \right)^{n^2} \leq \sum_{k=1}^{n^2} (n^2 \e^{-x^2+1} )^k \leq \frac{n^2 \e^{-x^2+1}}{1 - n^2 \e^{-x^2+1} }
    \leq 2\,n^2\,\e^{-x^2+1} .
  \end{equation}
  To satisfy Eq.~\eqref{eq:convergencecondition} for large $n$, it is sufficient that $x$ grows slightly faster than $\sqrt{\log(n^2)}$ and we hence need to demand a growth slightly faster than $\log(n^2)$ from the argument of the square root in the error function in Eq.~\eqref{eq:boundonlargeprobabilitystillwitherfc}.
  Because of the bound $n! \leq \e^{1-n}\,n^{n+1/2}$ (a variant of Stirling's approximation) we have for the argument of that square root in Eq.~\eqref{eq:boundonlargeprobabilitystillwitherfc}
  \begin{align}\label{eq:boundfortheroot}
    \frac{c}{2} \frac{\epsilon^{1/n} n^{\nu} }{(n!)^{2/n} } \geq \frac{c}{2} \frac{\epsilon^{1/n} n^{\nu} }{\e^{2/n-2} n^{2+1/n} } = \frac{c}{2} \frac{\epsilon^{1/n}}{\e^{2/n-2}} n^{\nu-2-1/n} ,
  \end{align}
  Demanding $\nu > 2$ is hence all we need to be able to use the bound \eqref{eq:boundbygeometricseries} for large $n$.
  With the convenient choice $\epsilon = 2^{-2n}$ it hence follows that for all $\nu>2$
  \begin{equation}
  \label{eq:finalprobabilitybound}
    \begin{split}
      &\Pr_{U \sim \mu_H}\left[ \exists S\in\Phi_{m,n}: \pbos(S) \geq 2^{-2n} \right] \\
      &\in O\left(n^2\,(2\,(c+1)\,\e)^n\,n^{(\nu-1)n} \exp(-c\,\e^{-2/n+2}\,n^{\nu-2-1/n} /8) \right) .
    \end{split}
  \end{equation}
  The argument of the $O(\cdot)$ is dominated by the product $n^{(\nu-1)n} \exp(-c\,\e^{-2/n+ 2}\,n^{\nu-2-1/n} /8)$, which decays for large increasing $n$ only for $\nu>3$.
  More precisely, there are constants $n_0 \in \mb N$ and $C_1,C_2,C_3 > 0 $ such that for $n \ge n_0$
  \begin{align}
     n^2\,& (2\,(c+1)\,\e)^n\,n^{(\nu-1)n} \exp(-c\,\e^{-2/n+2}\,n^{\nu-2-1/n} /8)  \\
    &  = \exp\left( 2\ln(n) + n \ln(2\,(c+1)\,\e) + n(\nu - 1) \ln n -c\,\e^{-2/n+2}\,n^{\nu-2-1/n} /8 \right)\\
    & \leq   \exp\left( C_1 n(\nu - 1) \ln n -c\,\e^{-2/n+2}\,n^{\nu-2-1/n} /8) \right)\\
    & \leq   \exp \left(C_1 n ( \nu - 1) \ln n - C_2  n^{\nu - 2 - 1/n} \right)\\
    & 
    \stackrel{\nu > 3}{\leq}  \exp\left(- C_3 n^{\nu - 2 - 1/n} \right )  \in \exp\left( - \Omega(n^{\nu - 2 - 1/n})\right).
  \end{align}
  where the last inequality holds only for $\nu > 3$ since the logarithm grows slower than any power law with positive exponent. 
  This completes the proof.

\end{proof}

\ifjournal
\putbib
\end{bibunit}
\else
  \twocolumngrid
  \bibliographystyle{./myapsrev4-1}
  \bibliography{sample_complexity_supremacy}
\fi

\end{document}